\documentclass[a4paper,USenglish]{lipics-v2021}

\hideLIPIcs  

\usepackage[noend]{algpseudocode}
\usepackage{amsmath,amsfonts,amsthm,algorithm}
\usepackage{wrapfig}
\usepackage{graphicx}
\usepackage{textcomp}
\usepackage{xcolor}
\usepackage{tikz}
\usepackage{multicol}
\usepackage{url,subcaption,svg,tcolorbox,arydshln}
\usetikzlibrary{decorations.pathreplacing}

\newcommand{\indistinguish}[1]{\stackrel{#1}{\sim}}

\algdef{SE}{Upon}{EndUpon}[1]{\textbf{Upon}\ #1\ \algorithmicdo}{}
\algtext*{EndUpon}

\theoremstyle{plain}

\theoremstyle{remark}

\usepackage{comment}  

\usepackage{ulem}
\newcommand{\lewis}[1]{{\color{red} #1}}

\newcommand{\PP}{\mathop{\mathbb{P}}}

\newcommand{\calA}{\mathcal{A}}
\newcommand{\calC}{\mathcal{C}}

\newcommand{\collect}[0]{\textsc{Collect}}
\newcommand{\store}[0]{\textsc{Store}}

\title{The Power of Abstract MAC Layer: A Fault-tolerance Perspective} 

\author{Qinzi Zhang}{Boston University, USA}{qinziz@bu.edu}{https://orcid.org/0000-0002-2243-5107}{}

\author{Lewis Tseng}{UMass-Lowell, USA}{lewistseng@acm.org}{https://orcid.org/0000-0002-4717-4038}{This material is based upon work partially supported by the National Science Foundation under Grant CNS-2334021.}

\authorrunning{Qinzi Zhang and Lewis Tseng} 

\Copyright{Jane Open Access and Joan R. Public} 



\category{} 

\relatedversion{} 

\acknowledgements{The authors want to thank the anonymous reviewers for insightful comments. The work was partially done when Lewis Tseng was affiliated with Boston College and Clark University.}

\nolinenumbers 

\EventEditors{John Q. Open and Joan R. Access}
\EventNoEds{2}
\EventLongTitle{42nd Conference on Very Important Topics (CVIT 2016)}
\EventShortTitle{CVIT 2016}
\EventAcronym{CVIT}
\EventYear{2016}
\EventDate{December 24--27, 2016}
\EventLocation{Little Whinging, United Kingdom}
\EventLogo{}
\SeriesVolume{42}
\ArticleNo{23}

\begin{document}

\maketitle

\begin{abstract}
This paper studies the power of the ``\textit{abstract MAC layer}'' model in a single-hop asynchronous network. 
The model captures primitive properties of modern wireless MAC protocols. 
In this model, Newport [PODC '14] proves that it is impossible to achieve deterministic consensus when nodes may crash. Subsequently,  Newport and Robinson [DISC '18]  present randomized consensus algorithms that terminate with $O(n^3 \log n)$ expected broadcasts in a system of $n$ nodes. We are not aware of any results on other fault-tolerant distributed tasks in this model. 

We first study the computability aspect of the abstract MAC layer. We present a \textit{wait-free} algorithm that implements an atomic register. Furthermore, we show that in general, $k$-set consensus is impossible. 
Second, we aim to  minimize storage complexity. Existing algorithms require $\Omega(n \log n)$ bits. We propose four \textit{wait-free} consensus algorithms that only need  \textit{constant} storage complexity. (Two approximate consensus and two randomized binary consensus algorithms.) One randomized algorithm terminates with $O(n \log n)$ expected broadcasts. All our consensus algorithms are \textit{anonymous}, meaning that at the algorithm level, nodes do not need to have a unique identifier. 

\end{abstract}

\section{Introduction}
\label{s:intro}


This paper studies fault-tolerant primitives,  with the focus on the aspect of wireless links in a single-hop asynchronous network. We adopt  the ``\textit{abstract MAC layer}'' model \cite{AbstractMAC_randomizedConsensus_DISC2018,AbstractMAC_consensus_PODC2014,AbstractMAC_DC2011},  
which captures the basic properties guaranteed by existing wireless MAC (medium access control) layers such as TDMA (time-division multiple access) or CSMA (carrier-sense multiple access). Even though the abstraction does not model after any specific existing MAC protocol, 
the abstract MAC layer still serves  an important goal -- the separation of high-level algorithm design and low-level logic of handling the wireless medium  and managing participating nodes. This separation helps identify principles that fills the gap between theory and practice in designing algorithms that can be readily deployed onto existing MAC protocols \cite{AbstractMAC_randomizedConsensus_DISC2018,AbstractMAC_consensus_PODC2014}. In fact, recent works in the networking community propose approaches to implement the abstract MAC layer in more realistic network conditions, e.g., dynamic systems \cite{AbstractMAC_implementation_CSDynamic_TransNetwork2021}, dynamic SINR channels \cite{AbstractMAC_implementation_Dynamic_TranMobile2021}, and Rayleigh-Fading channels  \cite{AbstractMAC_implementation_Rayleigh-Fading_TransWireless2021}. 

Consider an asynchronous network \cite{welch_book,AA_nancy} in which messages may suffer an arbitrary delay. Compared to conventional  message-passing models \cite{welch_book,AA_nancy}, the abstract MAC layer has two key characteristics (formal definition in Section \ref{s:preliminary}):
\begin{itemize}
    \item Nodes use a broadcast primitive which sends a message to all nodes that have \textit{not} crashed yet, and triggers an  acknowledgement upon the completion of the broadcast.

    \item Nodes do \textit{not} have a priori information on other participating nodes. 
\end{itemize}


The second characteristic is inspired by the observation that in a practical large-scale deployment, it is difficult to configure and manage all the connected devices so that they have the necessary information about other nodes. This assumption makes it difficult to port algorithms from conventional message-passing models to the abstract MAC layer, as these algorithms typically require the knowledge of other nodes and/or the size of the system. In fact, Newport and Robinson prove \cite{AbstractMAC_randomizedConsensus_DISC2018} that in message-passing models it is impossible to solve deterministic and randomized consensus, even if there is no fault, and nodes are assumed to have a constant-factor approximation of the network size.

In the abstract MAC layer model, Newport \cite{AbstractMAC_consensus_PODC2014} proves that deterministic consensus is impossible  when nodes may crash. Subsequently, Newport and Robinson \cite{AbstractMAC_randomizedConsensus_DISC2018} propose randomized consensus algorithms.  We are not aware of any study on other fault-tolerant primitives. This paper answers the following two fundamental  problems: 

\begin{itemize}
    \item Can we implement other fault-tolerant primitives?

    \item How do we minimize storage complexity when designing fault-tolerant primitives? 
\end{itemize}





\vspace{3pt}
\noindent\textbf{First Contribution.} In Herlihy's wait-free hierarchy \cite{herl:wait}, the consensus number defines the ``power'' of a shared object (or primitive). An object has a consensus number $c$, if it is possible for $\leq c$ nodes to achieve consensus using the object and atomic registers, \textit{and} it is not possible for $c+1$ nodes to do so. For example, an atomic register has consensus number 1, whereas consensus and compare-and-swap have consensus number $\infty$. The proof in \cite{AbstractMAC_consensus_PODC2014} implies that any objects with consensus number $\geq 2$ cannot be implemented in the abstract MAC layer. The natural next step is to understand whether objects with consensus number 1 can be implemented in the abstract MAC layer. 

We first show that the abstract MAC layer is fundamentally related to the \textit{store-collect object} \cite{store-collect_Attiya_DC2002,Store-collect_churn_SSS2020} by presenting a simple \textit{wait-free} algorithm to implement the object in the abstract MAC layer.  ``Stacking'' the constructions in \cite{Store-collect_churn_SSS2020} on top of our store-collect object solves many well-known computation tasks, e.g., registers, counters, atomic snapshot objects, and approximate, and randomized consensus. That is, we provide a wait-free approach to implement \textit{\underline{some primitives with consensus number 1} in the abstract MAC layer.} 


Next, we identify that \textit{not all} primitives with consensus number 1 can be implemented. In particular, we prove that in a system of $n$ nodes, $(n-1)$-set consensus is impossible to achieve in the abstract MAC layer model. This implies that other similar objects, like write-and-read-next objects \cite{Gafni_Write-and-read-next_DISC2018}, cannot be implemented as well. 

\vspace{3pt}
\noindent\textbf{Second Contribution.}~
From a more practical perspective, we study \textit{anonymous} and \textit{storage-efficient} fault-tolerant primitives. First, \textit{anonymous} algorithms do not assume unique node identity, and thus lower efforts in device configuration and deployment. Second,  most wireless devices are made small; hence, naturally, they are not equipped with abundant storage capacity, and storage-efficiency is an important factor in practical deployment. 


\begin{table}[t]
\centering
\footnotesize
\begin{tabular}{l|ccl}
                  & \textbf{consensus } & \textbf{storage complexity} & \textbf{note}               \\\hline
\textbf{NR18}~\cite{AbstractMAC_randomizedConsensus_DISC2018}     & randomized             & $\Theta(n \log n)$ bits                        & $O(n^3\log n)$ expected broadcasts 
\\
\textbf{MAC-RBC}  & randomized        & 8 values, 4 Booleans        & $O(2^n)$ expected  broadcasts\\  
\textbf{MAC-RBC2} & randomized        & 12 values, 5 Booleans        & $O(n \log n)$ expected  broadcasts\\  
\textbf{MAC-AC}   & approximate         & 4 values, 1 Boolean         & convergence rate $1/2$  \\
\textbf{MAC-AC2}  & approximate         & 2 values, 1 Boolean         & convergence rate $1-2^{-n}$
\end{tabular}

\caption{\footnotesize \textbf{Consensus in abstract MAC layer.} 
\\$\bullet$ The bottom four rows present our algorithms, whereas NR18 is the algorithm from DISC '18 \cite{AbstractMAC_randomizedConsensus_DISC2018}. 
\\$\bullet$ 
For approximate consensus, the convergence rate identifies the ratio that the range of fault-free nodes' states decreases after each asynchronous round. The smaller the ratio, the faster the convergence. 
}
\label{t:MAC}
\vspace{-18pt}
\end{table}

Table \ref{t:MAC} compares the state-of-the-art algorithm NR18  \cite{AbstractMAC_randomizedConsensus_DISC2018} and our algorithms. All the randomized consensus algorithms work for binary inputs and all  algorithms are \textit{wait-free}. The time complexity is measured as the expected number of broadcasts needed for all fault-free nodes to output a value. 
For our algorithms, ``values'' can be implemented using an integer or a float data type in practice. The exact size of the values will become clear later. 

Due to space limits, we focus only on our randomized algorithms and present our approximate agreement algorithms along with the full analysis in Appendix \ref{app:approximate}. 


    
    
    

\section{Preliminary}
\label{s:preliminary}


\noindent\textbf{Related Work.}~
We first discuss prior  works in the abstract MAC layer model. The model is proposed by Kuhn, Lynch and Newport \cite{AbstractMAC_DC2011}. They present algorithms for multi-message broadcast, in which multiple messages may be sent at different times and locations in a multi-hop network communicating using the abstract MAC layer. Subsequent works \cite{AbstractMAC_LEMIS_FOMC2012,AbstractMAC_probabilistic_AdHoc2014,AbstractMAC_unreliableLink_PODC2014} focus on \textit{non-fault-tolerant} tasks, including leader election and MIS.



The closest works are by Newport \cite{AbstractMAC_consensus_PODC2014} and Newport and Robinson \cite{AbstractMAC_randomizedConsensus_DISC2018}. Newport presents several impossibilities for  achieving \textit{deterministic} consensus when nodes may crash  \cite{AbstractMAC_consensus_PODC2014}. 
Newport and Robinson \cite{AbstractMAC_randomizedConsensus_DISC2018} present a randomized consensus algorithm that terminates after $O(n^3 \log n)$ broadcasts w.h.p. In their algorithm, nodes need to count the number of acknowledgements received from unique nodes and determine when to safely output a value. As a result, their algorithm requires storage space $\Theta(n \log n)$ bits 
and the knowledge of identities to keep track of unique messages. An accompanied  (randomized) approach of assigning node identities with high
probability is also proposed in \cite{AbstractMAC_randomizedConsensus_DISC2018}. Tseng and Sardina \cite{Tseng_OPODIS23_ByzMAC} present Byzantine consensus algorithms in the abstract MAC layer model, but they assume the knowledge of an upper bound on $n$ and  unique identities. Our consensus algorithms do not rely on identities; hence, fundamentally use different techniques. 


Fault-tolerant consensus has been studied in various models that assume message-passing communication links  \cite{welch_book,AA_nancy}. 
We consider a different communication model; hence, the techniques are quite different. An important distinction is that with asynchronous message-passing, it is impossible to implement a \textit{wait-free} algorithm  \cite{herl:wait}. Furthermore, nodes require accurate information on the network size \cite{AbstractMAC_randomizedConsensus_DISC2018}. 

Di Luna et al. have a series of works on anonymous dynamic network \cite{Bonomi_DynamicAnonymousCounting_ICDCN14,Bonomi_DynamicAnonymousCounting_ICDCS14,DiLuna_DynamicAnonymous_DISC23,DiLuna_DynamicAnonymous_PODC23,DiLuna_DynamicAnonymous_PODC15,DiLuna_DynamicAnonymous_OPODIS15}. They do not assume any failures. A series of papers  \cite{CUP_DSN2007,CUP_OPODIS2008,CUP_AdHoc-Now2004} study a related problem, called consensus with unknown participants (CUPs), where nodes are only allowed to communicate with other nodes whose identities have been provided by some external mechanism. Our consensus algorithms do not need unique node identities. Failure detectors are used in \cite{Abboud_Consensus_NoN08,Ruppert_AnonymousConsensusFD_OPODIS07,AnonymousFD_DC2013} to solve consensus with anonymous nodes. We do not assume a failure detector.

\vspace{3pt}
\noindent\textbf{Model.}
We consider a static asynchronous system consisting of $n$ nodes, i.e.,  we do not consider node churn. Each node is assumed to have a unique identifier; hence, the set of nodes is also denoted as the set of their identifiers, i.e., $\{1, \dots, n\}$.  For brevity, we often denote it by $[n]$. 
Our construction of store-collect requires identifiers due to its semantics. Our approximate and randomized algorithms are anonymous, and do \textit{not} assume node identifiers. The identifiers are used only for analysis.

We consider the crash fault model in which \underline{any number of nodes may fail}. A faulty node may crash and stop execution at any point of time. The adversary may control faulty behaviors and the message delays. Nodes that are not faulty are called \textit{fault-free} nodes. 

In a single-hop network with abstract MAC layer  \cite{AbstractMAC_DC2011,AbstractMAC_consensus_PODC2014,AbstractMAC_randomizedConsensus_DISC2018}, nodes communicate using the \texttt{mac-broadcast} primitive, which eventually delivers the message to all the nodes that have not crashed yet, including $i$ itself. Moreover, at some point \underline{after} the \texttt{mac-broadcast} has succeeded in delivering the message, the broadcaster receives an acknowledgement, representing that the \texttt{mac-broadcast} is complete. The broadcaster \textit{cannot} infer any other information from the acknowledgment, not the system size $n$, nor the identities of other nodes. A crash may occur during the \texttt{mac-broadcast}, which leads to  inconsistency. That is, if a broadcaster crashes, then some nodes might receive the message while others do not.

The key difference between message-passing and abstract MAC layer models is that in the message-passing model, sender requires explicit responses, which is the main reason that this model does not support wait-free algorithms and requires a priori information on participating nodes. In other words, abstract MAC layer allows us to design primitives with stronger properties due to the implication from the acknowledgement. This paper is the first to identify how to use its ``power'' to implement (some) primitives  with consensus number 1.

For our store-collect and approximate consensus algorithms, the adversary may enforce an arbitrary schedule for message delivery and crashes. 
For our randomized algorithms, we assume the \textit{message oblivious (or value oblivious)} adversary as in \cite{AbstractMAC_randomizedConsensus_DISC2018}. That is, the adversary does \textit{not} know the private states of each node (process states or message content).


\label{sec:message-processing}

\vspace{3pt}
\noindent\textbf{Algorithm Presentation and Message Processing.} Each node is assumed to take steps sequentially (a single-thread process). Each line of the pseudo-code is executed \textit{atomically}, except when calling \texttt{mac-broadcast}, since this primitive is handled by the underlying abstract MAC layer. 
Each algorithm also has a message handler that processes incoming messages. Our algorithms assume that (i) the message handler is triggered whenever the underlying layer receives a message and sends an interrupt; (ii) there is only one message handler thread, which processes messages one by one, i.e., the underlying layer has a queue of pending messages; and (iii) the handler has a priority over the execution of the main thread. 
The third assumption implies the following observation, which is important for ensuring the correctness of our algorithms: 
\begin{remark}
\label{remark1}
At the point of time when the main thread starts executing a line of the pseudo-code, there is \textit{no} pending message to be processed by the handler.\footnote{This assumption is not needed in prior works \cite{AbstractMAC_consensus_PODC2014,AbstractMAC_randomizedConsensus_DISC2018}, because their algorithm design is fundamentally different from ours. On a high-level, their algorithms proceed in an atomic block, whereas our algorithms have shared variables between the main thread and multiple message handlers. The assumption captures the subtle interaction between them.} 
\end{remark}

It is possible that \textit{during} the execution of a line in the main thread, the underlying layer sends an interrupt. The message handler will process these messages \textit{after} the completion of that particular line of code due to the assumption of atomic execution. The only exception is the call to \texttt{mac-broadcast}. Messages can still be received and processed when a broadcaster is waiting for the acknowledgement from the abstract MAC layer. 
\newcommand{\collectTag}[0]{\texttt{Collect}}
\newcommand{\collectACK}[0]{\texttt{ACK-Collect}}
\newcommand{\storeTag}[0]{\texttt{STORE}}
\newcommand{\storeACK}[0]{\texttt{ACK-Store}}

\section{Abstract MAC Layer: Computability}
\label{s:computability}


From the perspective of computability, asynchronous point-to-point message-passing model is fundamentally related to linearizable shared objects \cite{hen:lin}. However, it was pointed out that register simulation in conventional point-to-point models like ABD \cite{AttiyaBD95} is ``thwarted'' \cite{AbstractMAC_randomizedConsensus_DISC2018}. In other words, this observation indicated that the computability of the abstract MAC layer remained an \textit{open} problem. We fill the gap by presenting a framework of implementing some linearizable shared objects with consensus number $1$. 


\vspace{3pt}
\noindent\textbf{Our Insight.}~In the point-to-point model, ABD requires the communication among a quorum, because ``information kept by a quorum'' ensures that the information is \textit{durable} and \textit{timely} in quorum-based fault-tolerant designs. Durable information means that the information is not lost, even after node crashes. 
Timely information means that the information satisfies the real-time constraint, i.e., after the communication with a quorum is completed, others can learn the information by contacting any quorum of nodes. 

Our important observation is that the \texttt{mac-broadcast} achieves \textit{both} goals upon learning the acknowledgement. That is, after the broadcaster learns that the broadcast is completed, it can infer that the message is both durable and timely. 

Durability and timeliness are indeed sufficient for ensuring ``regularity,'' which can then be used to implement linearizable shared objects (as will be seen in Theorem \ref{thm:sc}). We next present a construction of store-collect objects.

\vspace{3pt}
\noindent\textbf{Store-Collect Object.}~
A store-collect object \cite{store-collect_Attiya_DC2002,Store-collect_churn_SSS2020} provides two operations (or interfaces) at node $i$: (i) \store$_i(v)$: store value $v$ into the object; and (ii) \collect$_i{}()$: collect the set of ``most recent'' values (of the object) from each node. The returned value is a \textit{view} $V$ -- a set of $(v_j, j)$ tuples where $j$ is a node identity and $v_j$ is its most recent stored value. For each $j$, there is at most one tuple of $(*, j)$ in $V$. With a slight abuse of notation, $V(j) = v_j$ if $(v_j, j) \in V$; otherwise, $V(j) = \perp$.

To formally define store-collect, we first discuss a useful notion. A \textit{history} is an execution of the store-collect object, which can be represented using a partially ordered set $(H, <_{H})$. Here, $H$ is the set of invocation ($inv$) and response ($resp$) events of the \store{} and \collect{} operations, and $<_{H}$ is an irreflexive transitive relation that captures the real-time ``occur-before'' relation of events in $H$.  Formally, for any two events $e$ and $f$, we say $e <_{H} f$ if $e$ occurs before $f$ in the execution. For two operations $op_1$ and $op_2$, we say that $op_1$ precedes $op_2$ if $resp(op_1) <_{H} inv(op_2)$. 

Every value in \store{} is assumed to be  unique (this can be achieved using sequence numbers and node identifiers). A node can have at most one pending operation. Given views $V_1$ and $V_2$ returned by two \collect{} operations, we denote $V_1 \preceq V_2$, if for every $(v_1, j) \in V_1$, there exists a $v_2$ such that (i) $(v_2, j) \in V_2$; and (ii) either $v_1 = v_2$ or the invocation of \store$_j(v_2)$ occurs after the response of \store$_j(v_1)$. That is, from the perspective of node $j$, $v_2$ is more recent than $v_1$. We then say that a history $\sigma$ satisfies \textit{regularity} if: 

\begin{itemize}
    \item For each \collect{}() $c \in \sigma$ that returns $V$ and for each node $j$, (i) if $V(j) = \perp$, then no \store{} by $j$ precedes $c$ in $\sigma$; and (ii) if $V(j) = v$, then \store$_j(v)$'s invocation precedes $c$'s response, and there does not exist \store$_j(v')$ such that $v' \neq v$, and \store$_j(v')$'s response occurs after \store$_j(v)$'s response and before $c$'s invocation.
    
    
    
    
    \item Consider any pair of two \collect{}'s in history $\sigma$, $c_1$ and $c_2$, which  return views $V_1$ and $V_2$, respectively. If $c_1$ precedes $c_2$, then $V_1 \preceq V_2$. 
\end{itemize}

\begin{definition}[Store-Collect]
An algorithm correctly implements the store-collect object if every execution of the algorithm results into a history that satisfies regularity. 
\end{definition}



\vspace{3pt}
\noindent\textbf{Our Wait-free Construction of Store-Collect.}~
To achieve regularity, each stored value has to be durable and timely. If a value is not durable, then the first condition for regularity may be violated. If a value is not timely, then the second condition may be violated.
Moreover, any current information needs to be known by subsequent \collect's, potentially at other nodes. These observations together with the aforementioned insight of the \texttt{mac-broadcast} primitive give us a surprisingly straightforward construction. Our algorithm MAC-SC is presented in Algorithm \ref{alg:MAC-Store-Collect}. 

Each node $i$ keeps a local variable $view_i$, which is a set of values -- one value for each node (that is known to node $i$ so far). With a slight abuse of terminology, we use $C = A \cup B$ to denote the merge operation of two views $A$ and $B$, which returns a view $C$ that contains the newer value from each node. Since each node can have at most one pending operation and each value is unique, the notion of ``newer'' is well-defined. For brevity, the sequence number is omitted in the notation.  

For \store$_i{(v)}$, node $i$ first adds the value $v$ to form a new view, and uses \texttt{mac-broadcast} to inform others about the new view. Because this information is both durable and timely upon the completion of the broadcast, regularity is satisfied. Since the broadcast delivers the message to the broadcaster as well,  $(v,i)$ is added to $view_i$ at line 8. For \collect$_i{}()$, it is similar except that the broadcast view is the current local view at node $i$. Upon receiving a new view (from the incoming message with the \store{} tag), $i$ simply merges the new view and its local view $view_i$.

\begin{algorithm}[h]
\caption{MAC-SC: Steps at each node $i$}
\label{alg:MAC-Store-Collect}
\begin{algorithmic}[1]
\footnotesize
\item[{\bf Local Variable:} /* It can be accessed by any thread at $i$. */]{}
	
\item[] $view_i$ \Comment{view, initialized to $\emptyset$}
\vspace{-8pt}
\item[] \hrulefill
\vspace{-13pt}
\begin{multicols*}{2}
    \item[\textbf{When $\store_i{(v)}$ is invoked}:]
    \State $currentView_i \gets view_i \cup \{(v, i)\}$
    \State \texttt{mac-broadcast}($\storeTag{}, currentView_i$)
    \State \textbf{return} ACK \Comment{\store{} is completed}

    \vspace{3pt}

    \item[\textbf{When $\collect_i()$ is invoked}:]
    \State $currentView \gets view_i$
    \State \texttt{mac-broadcast}($\storeTag{}, currentView$)
    \State \textbf{return} $currentView$
    
    \columnbreak
    \vspace{3pt}
    
    \item[$\slash\slash$ \textit{Background message handler}]
    
    \Upon{receive($\storeTag{}, view$)} 
        \State $view_i \gets view_i \cup view$
    \EndUpon
    
\end{multicols*}
\vspace{-10pt}
\end{algorithmic}
\end{algorithm}



\begin{theorem}
\label{thm:sc}
MAC-SC implements the store-collect object.
\end{theorem}
\begin{proof}[Proof Sketch]
\textbf{\textit{Property I}}: Consider a \collect$_i{}()$ operation $c$ that returns $V$. For each node $j$, consider two cases:

\begin{itemize}
    \item $V(j) = \perp$: this means that node $i$ has not received any message from $j$'s \texttt{mac-broadcast}.  This implies that either \texttt{mac-broadcast} by $j$ is not yet completed, or node $j$ has not invoked any \texttt{mac-broadcast}. In both cases, no \store{}  by $j$ precedes $c$.
    
    \item $V(j) = v$: by construction, $v$ is in $V(j)$ because \store$_j{(v)}$ is invoked before $c$ completes. Next we show that there is no other \store{} by node $j$ that completes between two events: the response event of \store$_j{(v)}$ and the invocation of $c$. Assume by way of contradiction that \store$_j{(v')}$ completes between these two events. Now observe that: (i) By definition, \store$_j{(v)}$ precedes \store$_j{(v')}$, so $v'$ is more recent than $v$ from the perspective of node $j$; and (ii) By the assumption of the abstract MAC layer, when \store$_j{(v')}$ completes, node $i$ must have received the value $v'$. These two observations together imply that node $i$ will add $v'$ into its view at line 8 before the invocation of $c$. Consequently, $V(j) = v'$ in the view returned by $c$, a contradiction.
\end{itemize}

\noindent\textbf{\textit{Property II}}: Suppose $c_1$ and $c_2$ are two \collect{}'s such that $c_1$ returns view $V_1$,
$c_2$ returns view $V_2$, and $c_1$ precedes $c_2$. By assumption, when \texttt{mac-broadcast} completes, all the nodes that have not crashed yet have received the broadcast message. Therefore, $V_1 \preceq V_2$.
\end{proof}

\noindent\textbf{From Store-Collect to Linearizable Objects.}~
Constructions of several linearizable shared objects over store-collect are presented in \cite{Store-collect_churn_SSS2020}. These constructions only use \store{} and \collect{} without relying on other assumptions; hence, can be directly applied on top of MAC-SC. More concretely, Attiya et al. \cite{Store-collect_churn_SSS2020} consider a dynamic message-passing system, where nodes continually enter and leave. Similar to our model, their constructions do not assume any information on other participating nodes. All the necessary coordination is through the store-collect object. 

This stacked approach sheds light on the computability of the abstract MAC layer. We can use the approach in \cite{Store-collect_churn_SSS2020} to implement an atomic register on top of MAC-SC in abstract MAC layer in a \textit{wait-free} manner. Consequently, MAC-SC opens the door for the implementation of many shared objects with consensus number 1. In particular, any implementation on atomic register that does not require a priori information on participating nodes can be immediately applied, e.g., linearizable abort flags, sets, and max registers \cite{LatticeAgreement_Petr_OPODIS2019,Store-collect_churn_SSS2020}.

Interestingly, despite the strong guarantee, \textit{not all} objects with consensus number 1 can be implemented in the abstract MAC layer. In particular, in Appendix \ref{s:set-consensus}, 
we prove that $(n-1)$-set consensus is impossible to achieve. Our proof follows the structure of the counting-based argument developed by Attiya and Paz (for the shared memory model) \cite{Attiya_CoutingImpossibility_SetConsensus_DISC2012}.  
\section{Anonymous Storage-Efficient Randomized Binary Consensus}
\label{s:random-MAC}

While general, the stacked approach  comes with two drawbacks in practice -- assumption of unique identities and high storage complexity. 
Stacking prior shared-memory algorithms on top of MAC-SC requires $\Omega(n \log n)$ due to the usage of store-collect. Prior message-passing algorithms (e.g.,  \cite{AA_Dolev_1986,vaidya_PODC12,Ben-Or_PODC1983}) usually require the assumption of unique identities. 

This section considers anonymous storage-efficient  randomized binary consensus. Recall that deterministic consensus is impossible under our assumptions \cite{AbstractMAC_consensus_PODC2014}, so the randomized version is the best we can achieve. As shown in Table \ref{t:MAC}, the state-of-the-art algorithm NR18 \cite{AbstractMAC_randomizedConsensus_DISC2018} requires $O(n^3\log n)$ time complexity w.h.p. and $\Theta(n \log n)$ storage complexity. We present two \textit{anonymous wait-free} algorithms using only constant storage complexity. 

\vspace{3pt}
\noindent\textbf{Our Techniques.}~~ 
Our algorithms are inspired by Aspnes's  framework  \cite{RandConsensus_Aspnes_DC2012} of alternating adopt-commit objects and conciliator objects. The framework is designed for the shared memory model, requiring \textit{both} node identity and the knowledge of $n$. Moreover, it requires $O(\log n)$ atomic multi-writer registers in expectation.


To address these limitations, we have two key technical contributions. First, we replace atomic multi-writer registers by \texttt{mac-broadcast}, while using only constant storage complexity. Second, we integrate the ``doubling technique,'' for estimating the system size $n$, with the framework and present an accompanied analysis to bound the expected round complexity.

More concretely, we combine the technique from \cite{Tseng_DISC20_everyonecrash} and Aspnes's framework to avoid using new objects in a new phase. 
More precisely, we borrow the ``jump'' technique from \cite{Tseng_DISC20_everyonecrash}, which allows nodes to skip phases (and related messages), to reduce storage complexity. This comes with two technical challenges. First, our proofs are quite different from the one in \cite{Tseng_DISC20_everyonecrash}, because nodes progress in a different dynamic due to the characteristics of the abstract MAC layer. In particular, we need to carefully analyze which broadcast message has been processed to ensure that the nodes are in the right phase in our proof. This is also where we need to rely on Remark \ref{remark1}, which is usually not needed in the proofs for point-to-point message-passing models. 
Second, compared to \cite{RandConsensus_Aspnes_DC2012}, our proofs are more subtle in the sense that we need to make sure that concurrent broadcast events and ``jumps'' do not affect the probability analysis. The proof in \cite{RandConsensus_Aspnes_DC2012} mainly relies on the atomicity of the underlying shared memory, whereas our proofs need to carefully analyze the timing of  broadcast events. (Recall that we choose not to use MAC-SC, since it requires nodes to have  unique identities.)

Prior solutions rely on the knowledge of network size $n$ \cite{RandConsensus_firstMover_OPODIS2005,RandConsensus_Chor_firstMover_SIAM1994,RandConsensus_Aspnes_DC2012} or an estimation of $n$ \cite{AbstractMAC_randomizedConsensus_DISC2018} to improve time complexity.  For anonymous storage-efficient algorithms, nodes do not know $n$, and  there is no unique node identity. The solution for estimating the network size in the abstract MAC layer in \cite{AbstractMAC_randomizedConsensus_DISC2018} only works correctly with a large $n$ (i.e., with high probability). We integrate a ``doubling technique'' to \textit{locally} estimate $n$ which does not require any message exchange. For our second randomized binary consensus algorithm MAC-RBC2, nodes double the estimated system size $n'$ every $c$ phases for some constant $c$, if they have not terminated yet. We identify a proper value of $c$ so that $n'$ is within a constant factor of $n$, and nodes achieve agreement using $O(n \log n)$ broadcasts on expectation. 


\vspace{3pt}
\noindent\textbf{Randomized Binary Consensus and Adopt-Commit.}

\begin{definition}[Randomized Binary Consensus]
    A correct randomized binary consensus algorithm satisfies: (i) \textbf{Probabilistic Termination}: Each fault-free node decides an output value with probability 1 in the limit; (ii) \textbf{Validity}: Each output is some input value; and (iii) \textbf{Agreement}: The outputs are identical.
    
        
        
\end{definition}

\begin{definition}[Adopt-Commit Object]
    A correct adopt-commit algorithm satisfies: (i) \textbf{Termination}: Each fault-free node outputs either $(commit, v)$ or $(adopt, v)$ within a finite amount of time; (ii) \textbf{Validity}: The $v$ in the output tuple must be an input value; (iii) \textbf{Coherence}: If a node outputs $(commit, v)$, then any output is either $(adopt, v)$ or $(commit, v)$; and (iv) \textbf{Convergence}: If all inputs are $v$, then all fault-free nodes output  $(commit, v)$. 
    
        
        
        
\end{definition}

\subsection{Algorithm MAC-AdoptCommit}

\begin{algorithm}[t]
\caption{MAC-AdoptCommit Algorithm: Steps at each node $i$ with input $v_i$}
\label{alg:MAC-AdoptCommit}
\begin{algorithmic}[1]
\footnotesize
\item[{\bf Local Variables:} /* These can be accessed by any thread at $i$. */]{}
	
\item[] $seen_i[0]$ \Comment{Boolean, initialized to $false$}
\item[] $seen_i[1]$ \Comment{Boolean, initialized to $false$}
\item[] $proposal_i$ \Comment{value, initialized to $\perp$}
\vspace{-5pt}
\item[] \hrulefill    
\vspace{-15pt}
\begin{multicols}{2}
    \State $\texttt{mac-broadcast}(\texttt{VALUE}, v_i)$
    \If{$proposal_i \neq \perp$}
        \State $v_i\gets proposal_i$
    \EndIf
    \State $\texttt{mac-broadcast}(\texttt{PROPOSAL}, v_i)$
    \If{$seen_i[-v_i]$ = false}\label{line:decide}
        \State \textbf{return} $(commit, v_i)$ 
    \Else    
        \State \textbf{return} $(adopt, v_i)$
    \EndIf
    \columnbreak
    \item[$\slash\slash$ \textit{Background message handler}]
    \Upon{receive($\texttt{VALUE}, v$)} 
        \State $seen_i[v] \gets true$
    \EndUpon
    \item[]
    \Upon{receive($\texttt{PROPOSAL}, v$)} 
        \State $proposal_i \gets v$
    \EndUpon
\end{multicols}
\vspace{-10pt}
\end{algorithmic}
\end{algorithm}

We present MAC-AdoptCommit, which implements a \textit{wait-free} adopt-commit object for binary inputs in the abstract MAC layer model. The pseudo-code is presented in Algorithm \ref{alg:MAC-AdoptCommit}, and the algorithm is inspired by the construction in shared memory \cite{AnonymousAdoptCommit_Aspnes_SPAA2011}. Following the convention, we will use $-v$ to denote the opposite (or complement) value of value $v$.

Each node $i$ has two Booleans, $seen_i[0]$ and $seen_i[1]$, and a value $proposal_i$. The former variables are initialized to $false$, and used to denote whether a node $i$ has seen input value $0$ and $1$, respectively. The last variable $proposal_i$ is initialized to $\perp$, and used to record the ``proposed'' output from some node. The algorithm has two types of messages: 

\begin{itemize}
    \item A \texttt{VALUE} type message $(\texttt{VALUE}, v_i)$ that is used to exchange input values.
    
    \item A \texttt{PROPOSAL} type message $(\texttt{PROPOSAL}, v_i)$ that is to announce a proposed value.
    
\end{itemize}
Upon receiving the message $(\texttt{VALUE}, v)$, node $i$ updates $seen_i[v]$ to $true$ (line 11), denoting that it has seen the value $v$. Upon receiving the message  $(\texttt{PROPOSAL}, v)$, $i$ updates $proposal_i$ to $v$ (line 13), denoting that it has recorded the proposed value, by either itself or another node. Due to concurrency and asynchrony, it is possible that there are multiple proposal messages; thus, node $i$ may overwrite existing value in $proposal_i$ with an opposite value. 

Node $i$ first broadcasts input $v_i$. After \texttt{mac-broadcast} completes (line 1), $i$ checks whether it has received any \texttt{PROPOSAL}  message. If so, it updates its state $v_i$ to the value (line 5). Otherwise, it becomes a proposer and broadcast $\texttt{PROPOSAL}$ message with its own input $v_i$ (line 3). After Line 5, the state $v_i$ could be $i$'s original input, or a state copied from the proposed value (from another node). Finally, if node $i$ has not observed any \texttt{VALUE} message with the opposite state ($-v_i$), then it outputs $(commit, v_i)$; otherwise, it outputs  $(adopt, v_i)$.


\vspace{3pt}
\noindent\textbf{Correctness.}~~
Validity, termination and convergence are obvious. To see how MAC-AdoptCommit achieves coherence, first observe that it is impossible for some node to output $(commit, v)$, and the others to output $(commit, -v)$. It is due to the property of \texttt{mac-broadcast}: if some node outputs $(commit, v)$, then every node must observe $seen[v] = true$ when executing line 6. Second, if a node outputs $(commit, v)$, then it must be the case that there has already been a proposer that has broadcast both message $(\texttt{VALUE}, v)$  and message $(\texttt{PROPOSAL}, v)$. Therefore, it is impossible for a node to output $(adopt, -v)$. For completeness, we present the proof of correctness in Appendix \ref{app:adoptCommit}.

\subsection{Algorithm MAC-RBC}

\begin{algorithm*}
\caption{MAC-RBC Algorithm: Steps at each node $i$ with input $x_i$}
\label{alg:MAC-RBC}
\begin{algorithmic}[1]
\footnotesize
\item[{\bf Local Variables:} /* These variables can be accessed and modified by any thread at node $i$. */]{}
	
\item[] $seen_i[0]$ \Comment{(Boolean, phase), initialized to $(false, 0)$}
\item[] $seen_i[1]$ \Comment{(Boolean, phase), initialized to $(false, 0)$}
\item[] $seen_i^2[0]$ \Comment{(Boolean, phase), initialized to $(false, 0)$}
\item[] $seen_i^2[1]$ \Comment{(Boolean, phase), initialized to $(false, 0)$}
\item[] $v_i$ \Comment{state, initialized to $x_i$, the input at node $i$}
\item[] $p_i$ \Comment{phase, initialized to $0$}
\item[] $proposal_i$ \Comment{(value, phase), initialized to $(\perp,0)$}
\vspace{-5pt}
\item[] \hrulefill    
\vspace{-15pt}
\begin{multicols}{2}
    \While{true}
        \State $p_{old} \gets p_i$
        \State $\texttt{mac-broadcast}(\texttt{VALUE}, v_i, p_i)$
        
                
        
                
            \If{$proposal_i.phase \geq p_i$}
                \State $(v_i, p_i) \gets proposal_i$
            \EndIf
            \State $\texttt{mac-broadcast}(\texttt{PROPOSAL}, v_i, p_i)$
            
            \If{$p_{old} \neq p_i$}
                \State \textbf{go to} line 2 in  $p_i$ \Comment{``Jump'' to $p_i$}
            \ElsIf{$seen_i[-v_i].phase < p_i$} 
                \State \textbf{output} $v_i$
            \Else    
                \State \texttt{mac-broadcast}($\texttt{VALUE}^2, v_i, p_i$)
                \If{$seen_i^2[-v_i].phase > p_i$}
                    \State $(v_i, p_i) \gets seen_i^2[-v_i]$
                    \State \textbf{go to} line 2 in  $p_i$ \Comment{``Jump'' to $p_i$}
                \ElsIf{$seen_i^2[-v_i]=(true, p_i)$}
                        \State $v_i \gets$\textsc{FlipLocalCoin()}
                \EndIf    
                \State $p_i \gets p_i + 1$ \Comment{``Move'' to $p_i$}
            \EndIf
    \EndWhile
    
    \columnbreak
    
    \item[$\slash\slash$ \textit{Background message handler}]
    \Upon{receive($\texttt{VALUE}, v, p$)} 
        \If{$p \geq p_i$}
            \State $seen_i[v] \gets (true, p)$
        \EndIf    
    \EndUpon
    \item[]
    \Upon{receive($\texttt{VALUE}^2, v, p$)} 
        \If{$p \geq p_i$}
            \State $seen_i^2[v] \gets (true, p)$
        \EndIf    
    \EndUpon
    \item[]
    \Upon{receive($\texttt{PROPOSAL}, v, p$)} 
        \If{$p \geq p_i$}
            \State $proposal_i \gets (v, p)$
        \EndIf
    \EndUpon
\end{multicols}
\vspace{-10pt}
\end{algorithmic}
\end{algorithm*}

We present MAC-RBC in Algorithm \ref{alg:MAC-RBC}. The algorithm uses a sequence of adopt-commit  and conciliator objects. A conciliator object helps nodes to reach the same state, and an adopt-commit is used to determine whether it is safe to output a value, and choose a value for the next phase when one cannot ``commit'' to an output. We adapt MAC-AdoptCommit to store phase index, which allows nodes to jump to a higher phase. Effectively, any adopt-commit object with a phase $< p$ can be interpreted as having $\perp$ in phase $p$. This also allows us to ``reuse'' the object. For the conciliator object, we use Ben-Or's local coin \cite{Ben-Or_PODC1983}, which achieves expected exponential time complexity. 

In Algorithm \ref{alg:MAC-RBC}, line 3 to line 10 effectively implement a reusable adopt-commit object using $\texttt{VALUE}$ and $\texttt{PROPOSAL}$ messages. Line 12 to line 17 implement a conciliator object using the $\texttt{VALUE}^2$ message. The $seen$ variables store both a Boolean and a phase index. Nodes only update these variables when receiving a corresponding message from the same or a higher phase. A node $i$ flips a local coin to decide the state for the next phase at line 17 \textit{only if} it can safely infer that both $0$ and $1$ are some node's state at the beginning of the phase, i.e., it flips a coin when it has \textit{not} seen a $\texttt{VALUE}^2$ message from a higher phase (line 13), and it has observed a $\texttt{VALUE}^2$ message with value $-v_i$ from the same phase (line 16).

\vspace{3pt}
\noindent\textbf{Correctness Proof.}~
It is straightforward to see that MAC-RBC satisfies validity, since the state is either one's input or a value learned from received messages (which must be an input value) and there is no Byzantine fault. We then prove the agreement property.

\begin{lemma}
Suppose node $i$ is the first node that makes an output and it outputs $v$ in phase $p$, then all the other nodes either output $v$ in phase $p$ or phase $p+1$.
\end{lemma}

\begin{proof}
Suppose node $i$ outputs $v$ in phase $p$ at time $T_3$. Then it must have $seen_i[-v].phase < p$. Assume this holds true at time $T_2$. Furthermore, assume line 6 was executed at time $T_1$ by node $i$ at time $T_1$ such that $T_1 < T_2 < T_3$.

We first make the observation, namely \textit{Obs1}, no node with $-v$ in phase $p' \geq p$ has completed line 3 at time $\leq T_2$. Suppose node $j$ has state $-v$ in some phase $p' \geq p$. By assumption (in Section \ref{sec:message-processing}), before node $i$ starts to execute line 9 at time $T_2$, its message handler has processed all the messages received by the abstract MAC layer. Therefore, the fact that $seen_i[-v].phase < p_i$ at time $T_2$ implies that node $i$ has \textit{not} receive any message of the form $(\texttt{VALUE}, -v, p')$ at time $T_2$. Consequently, node $j$ has not completed \texttt{mac-broadcast}(\texttt{VALUE}, $-v, p'$) (line 1) at time $T_2$.

Consider the time $T$ when the first \texttt{mac-broadcast}(\texttt{VALUE}, $-v, p$) is completed (if there is any). At time $T$, there are two cases for node $k$ that has not crashed yet: 

\begin{itemize}
    \item Node $k$ has \textit{not} moved beyond phase $p$: 
    
    $k$ must have already received (\texttt{PROPOSAL}, $v$) at some earlier time than $T$, because (i) \textit{Obs1} implies that $T > T_2$; and (ii) by time $T$, node $i$ has already completed line 6 (which occurred at time $T_1$). Consider two scenarios: (s1) $k$ executes line 4 after receiving (\texttt{PROPOSAL}, $v$): $k$ sets $proposal_k$ to value $v$ before executing line 6 (potentially at some later point than $T$); and (s2) $k$ executes line 4 before receiving (\texttt{PROPOSAL}, $v$): in this case: $k$'s input at phase $p$ must be $v$; otherwise, $T$ cannot be the first \texttt{mac-broadcast}(\texttt{VALUE}, $-v, p'$) that is completed. (Observe that by assumption of this case, $k$ executes line 4 before node $i$ completes its line 6 at time $T_1$.) 

    
    
    \item Node $k$ has moved beyond phase $p$: 
    
    By assumption, time $T$ is the time that the first \texttt{mac-broadcast}(\texttt{VALUE}, $-v, p$)  is completed. Thus, it is impossible for node $k$ to have set $(v_k, p_k)$ to $(-v, p')$ for some $p' \geq p$.
\end{itemize}
In both cases, right before executing line 6, node $k$ can only \texttt{mac-broadcast}(\texttt{PROPOSAL}, $v, p'$), for $p' \geq p$, i.e., no \texttt{mac-broadcast}(\texttt{PROPOSAL}, $-v, p'$) is possible. Consequently, the lemma then follows by a simple induction on the order of nodes moving to phase $p+1$.  
\end{proof}

Since we assume a message oblivious adversary, the termination and exponential time complexity follow the standard argument of using local coins \cite{Ben-Or_PODC1983}.
In particular, we have the following Theorem, which implies that MAC-RBC requires, on expectation, an exponential number of broadcasts. The proof is deferred to Appendix \ref{app:MAC-RBC-termination}.

\begin{theorem}
\label{lem:MAC-RBC-termination}
For any $\delta\in(0,1)$, let $p=\lceil2^{n-1}\ln(1/\delta)\rceil$. Then with probability at least $1-\delta$, MAC-RBC terminates  within $p$ phases. (In other words, all nodes have phase $\le p$.)
\end{theorem}

\subsection{MAC-RBC2: Improving Time Complexity}



There are several solutions for an  efficient conciliator object, such as a shared coin \cite{aspnes1990wait} and the ``first-mover-win'' strategy \cite{RandConsensus_firstMover_OPODIS2005,RandConsensus_Chor_firstMover_SIAM1994,RandConsensus_Aspnes_DC2012}. 
The first-mover-win strategy was developed for a single multi-writer register in shared memory such that  agreement is achieved when only one winning node (the first mover) successfully writes to the register. If there are concurrent operations, then agreement might be violated. On a high-level, this strategy translates to the ``first-broadcaster-win'' design in the abstract MAC layer. One challenge in our analysis is the lack of the atomicity of the register. We need to ensure that even in the presence of concurrent broadcast and failure events,  there is still a constant probability for achieving agreement, after nodes have a ``good enough'' estimated system size $n'$. 


\vspace{3pt}
\noindent\textbf{Conciliator and Integration.}~~
Our conciliator object is presented in Algorithm \ref{alg:MAC-FirstMover}, which is inspired by the ImpatientFirstMover strategy  \cite{RandConsensus_Aspnes_DC2012}. The key difference from \cite{RandConsensus_Aspnes_DC2012} is that MAC-FirstMover uses an \underline{estimated size $n'$}, instead of the actual network size $n$ (as in \cite{RandConsensus_Aspnes_DC2012}), which makes the analysis  more complicated, as our analysis depends on both $n$ and $n'$. Algorithm \ref{alg:MAC-FirstMover} presents a standalone conciliator implementation. We will later describe how to integrate it with Algorithm \ref{alg:MAC-RBC} by adding the field of phase index and extra message handlers.

In our design, each node proceeds in rounds and increases the probability of revealing their coin-flip after each round $k$, if it has not learned any coin flip at line 2. To prevent the message adversary from scheduling concurrent messages with conflicting values, nodes have two types of messages: \texttt{COIN} and \texttt{DUMMY}. The first message is used to reveal node's input $v_i$, whereas the second is used as a ``decoy'' that has no real effect. At line 3, node $i$ draws a local random number between $[0, 1)$ to decide which message to broadcast. Since the adversary is oblivious, it cannot choose its scheduling based on the message type. 

\begin{algorithm*}[t]
\caption{MAC-FirstMover Algorithm: Steps at each node $i$ with input $v_i$}
\label{alg:MAC-FirstMover}
\begin{algorithmic}[1]
\footnotesize
\item[{\bf Local Variables:} /* These variables can be accessed by any thread at node $i$. */]{}

\item[] $coin_i$ \Comment{value, initialized to $\perp$}

\item[{\bf Input:}]{$n'$} \Comment{estimated system size, given as an input to MAC-FirstMover}
\item[] \hrulefill    
\vspace{-15pt}
\begin{multicols}{2}
    \State $k \gets 0$
    \While{$coin_i = \perp$}
        \If{a local random number $< \frac{2^k}{2 n'}$}
            \State $\texttt{mac-broadcast}(\texttt{COIN}, v_i)$
        \Else
            \State $\texttt{mac-broadcast}(\texttt{DUMMY})$
        \EndIf
        \State $k \gets k+1$
    \EndWhile
    \State $\texttt{mac-broadcast}(\texttt{COIN}, coin_i)$
    \State \textbf{return} $coin_i$
    \columnbreak
    \item[$\slash\slash$ \textit{Background message handler}]
    \Upon{receive($\texttt{COIN}, v$)} 
        \If{$coin_i = \perp$}
            \State $coin_i \gets v$
        \EndIf
    \EndUpon
    \item[]
    \Upon{receive($\texttt{DUMMY}$)} 
        \State \textbf{do} nothing
    \EndUpon
\end{multicols}
\vspace{-10pt}
\end{algorithmic}
\end{algorithm*}

MAC-RBC2 can be obtained by integrating Alg. \ref{alg:MAC-FirstMover} (MAC-FirstMover) with Alg. \ref{alg:MAC-RBC} (MAC-RBC) with the changes below. The complete algorithm is presented in Appendix \ref{app:algorithm-MAC-RBC2}. 

\begin{itemize}
    \item \textsc{FlipLocalCoin}() is replaced by \textsc{MAC-FirstMover}($2^{\lfloor\frac{p_i}{c}\rfloor} n_0$), where $c$ is a constant to be defined later 
    and $n_0$ is a constant that denotes the initial guess of the system size. All nodes have an identical information of $c$ and $n_0$ in advance. Therefore, nodes in the same phase call MAC-FirstMover with the same estimated system size $n'$.  Recall that $p_i$ is the phase index local at node $i$. Hence, effectively in our design, each node $i$ is doubling the estimated size $n'$ every $c$ phases. 

    \item To save space, $coin_i$ consists of two fields $(value, phase)$, and is used in a fashion similar to how $proposal_i$ is used in MAC-RBC. That is, if $coin_i$ has a phase field lower than the current phase $p_i$, then the value field is treated as $\perp$. 

    \item The messages by node $i$ are tagged with its current phase $p_i$. That is, the two messages in Algorithm \ref{alg:MAC-FirstMover} have the following form: $(\texttt{COIN}, v, p_i)$ and $(\texttt{DUMMY}, p_i)$.

    \item MAC-RBC2 needs to have two extra message handlers to process \texttt{DUMMY} and \texttt{COIN} messages. The \texttt{COIN} message handler only considers messages with phase $\leq p_i$. 

    \item In MAC-RBC2, nodes jump to a higher phase upon receiving a coin broadcast. More precisely, if a node $i$ receives a coin broadcast $m$ from a phase $p > p_i$, then $i$ updates $v_i$ to the value in $m$ and jumps to phase $p+1$. 
\end{itemize}

\vspace{3pt}
\noindent\textbf{Correctness and Time Complexity}.~Correctness follows from the prior correctness proof, as MAC-FirstMover is a valid conciliator object that returns only $0$ or $1$. To analyze time complexity, we start with several useful notions.

\begin{definition}[Active Nodes]
We say a node is an \underline{\textit{active} node in phase $p$} if it ever executes MAC-FirstMover in phase $p$. Let $\calA_p$ denote the set of all active nodes in phase $p$.
\end{definition}
Due to asynchrony, different nodes might execute MAC-FirstMover in phase $p$ at different times. Moreover, nodes may ``jump'' to a higher phase in our design. Consequently, not all nodes would execute MAC-FirstMover in phase $p$ for every $p$. 

\begin{definition}[Broadcast]We distinguish different types of broadcasts, which will later be useful for our probability analysis: 

    \begin{itemize}
        \item A broadcast is a \textit{phase-$p$ broadcast} if it is tagged with phase $p$. By definition, only active nodes in phase $p$ (i.e., nodes in $\calA_p$) make phase-$p$ broadcasts.
        
        \item A broadcast made in MAC-FirstMover (Algorithm \ref{alg:MAC-FirstMover}) is a \underline{\textit{coin broadcast}} if its message has the \texttt{COIN} tag; otherwise, it is a \underline{\textit{dummy broadcast}}. 

        \item A broadcast is an \underline{\textit{original broadcast}} if it is made in the while loop (Line 4 and Line 6 in Algorithm \ref{alg:MAC-FirstMover}). It is a \underline{\textit{follow-up broadcast}} if it is made after $coin_i$ becomes non-empty (line 8 of Algorithm \ref{alg:MAC-FirstMover}). By design, a follow-up broadcast must be a coin broadcast. 
        
        \item Consider an original broadcast $m = (\texttt{COIN}, v, p)$ by node $i$. The broadcast is said to be \underline{\textit{successful in phase $p$}} if there exists a node $j$ that completes a follow-up broadcast with $coin_j = v$ in phase $p$, i.e., node $j$ receives the acknowledgement for its broadcast at line 8 of Algorithm \ref{alg:MAC-FirstMover}. Note that $i$ may not equal to $j$, and both $i$ and $j$ might be faulty (potentially crash at a future point of time). 
    \end{itemize}
    
\end{definition}

Recall that we define a broadcast to be ``\textit{completed}'' if a node making the broadcast receives the acknowledge from the abstract MAC layer. This notion should not be confused with the notion of ``successful.'' In particular, we have (i) a broadcast might be completed, but not successful -- this is possible if there are multiple original coin broadcasts with different $v$; (ii) a broadcast might be successful, but not completed -- this is possible if a node $j$ receives an original coin broadcast by a faulty node and node $j$ completes the follow-up broadcast. 

We will apply the following important observation in our proofs. The observation directly follows from our definition of different broadcasts.
\begin{remark}
\label{rmk:broadcast-1}
If there is a completed original coin broadcast in phase $p$, then there must be at least one successful original coin broadcast in phase $p$.
\end{remark}

\begin{definition}
    A node \textit{completes} MAC-FirstMover of phase $p$ if it receives a coin broadcast of the form $(\texttt{COIN}, *, p')$ with $p' \geq p$.\footnote{This coin broadcast can be an original or a follow-up coin broadcast.} Let $\calC_p$ denote the set of all nodes that complete MAC-FirstMover of phase $p$. 
\end{definition}
By definition, a node \textit{not} in $\calA_p$ can still complete MAC-FirstMover of phase $p$, if it receives a coin broadcast from a higher phase. 

We first bound the number of expected original broadcasts in order for nodes to complete MAC-FirstMover. Recall that $k$ in Algorithm \ref{alg:MAC-FirstMover} denotes the round index. In our analysis below, we only bound the number of broadcasts made by \textit{fault-free} nodes.

\begin{lemma}
\label{lemma:termination-firstMover}
    With probability $\ge 1-\delta$, all fault-free nodes complete MAC-FirstMover in phase $p$, after $\leq 2n' \ln(1/\delta)$ original broadcasts are made by fault-free nodes in phase $p$. 
\end{lemma}

\begin{proof}
We begin with the following claim. It follows from the definition of successful coin broadcasts. For completeness, we include the proof 
 in Appendix \ref{app:claim-one-coin-broadcast}.

\begin{claim}
\label{claim:one-coin-broadcast}
All fault-free nodes complete MAC-FirstMover of phase $p$ if there exists \underline{\textit{at least one} successful coin broadcast} in phase $p$. 
\end{claim}

Every broadcast in phase $p$ has probability $\ge \frac{1}{2n'}$ to be a coin broadcast (by line 3 of Algorithm \ref{alg:MAC-FirstMover}). Since we only care about the number of original broadcasts made by fault-free nodes, all these broadcasts must be eventually completed. Consequently, for all the fault-free nodes in $\calA_p$, we have the probability that \textit{first $t$ completed broadcasts by any fault-free node in $\calA_p$ are all dummy}, denoted by $P$, bounded by
\begin{align*}
    P
    &\le \prod_{i=1}^t \left(1 - \frac{1}{2n'}\right) 
    \le \exp\left( - \frac{t}{2n'} \right).
\end{align*}
Equivalently, for any $t\ge 2n'\ln(1/\delta)$, with probability at least $1-\delta$, there exists at least one completed coin broadcast among the first $t$ completed broadcasts in phase $p$, which further implies the existence of at least one successful broadcast by Remark \ref{rmk:broadcast-1}. 
This, together with Claim \ref{claim:one-coin-broadcast}, conclude  the proof. (Note that there could be a successful coin broadcast by a faulty node in $\calA_p$, but this does not affect the lower bound we derived.)
\end{proof}

\begin{lemma}
\label{lemma:agreement-firstMover}
    Consider the case when all active nodes in phase $p$ (i.e., nodes in $\calA_p$) execute MAC-FirstMover of phase $p$ with parameter $n' \ge n$. With probability $\ge 0.05$, each node $j \in \calC_p$ must reach the same state $v_j$ in either phase $p$ or phase $p+1$.
\end{lemma}

\begin{proof}

    

    We begin with the following claim. The proof is presented in Appendix \ref{app:claim:exactly-one-coin-broadcast}.

    \begin{claim}
        \label{claim:exactly-one-coin-broadcast}

        If there is exactly one successful original coin broadcast in phase $p$, then all nodes in $\calC_p$ must achieve the same state in either phase $p$ or phase $p+1$.
    \end{claim}

    The analysis below aims to identify the lower bound on the probability of the event that there exists exactly one successful original coin broadcast in phase $p$. 

    Consider any message scheduling by the adversary. Since we assume it is oblivious, we can define $r_i$ as the probability that the $i$-th completed original broadcast in phase $p$, across the entire set of nodes in $\calA_p$, is a coin broadcast given this unknown message scheduling. That is, since the schedule by the adversary is chosen a priori, $r_i$ is a fixed number. 
    Next, we introduce two variables:

    \begin{itemize}
        \item Let $T-1$ denote the \underline{number of completed original dummy broadcasts in phase $p$} before the first \textit{completed} original coin broadcasts in phase $p$, given the message scheduling; and

        \item Let $k_j$ denote the number of completed original dummy broadcasts by a node $j \in \calA$, among these $T-1$ broadcasts. Note that only $k_j$ is defined with respect to a single node.

    \end{itemize}
    The first definition implies that the $T$-th completed original broadcast is a coin broadcast. 

    Without loss of generality, assume that in the given schedule, the $i$-th completed original broadcasts across the entire set of nodes in $\calA_p$ is the $k$-th completed original broadcast made by node $j$. Then by Line 3 of Algorithm \ref{alg:MAC-FirstMover}, we can quantify $r_i$ as follows: 

    \begin{equation}
    \label{eq:ri}
    r_i = \frac{2^{k-1}}{2n'} 
    \end{equation}
    
    Observe that if some node $j\in\calA_p$ fails to complete an original broadcast, then it cannot make any further broadcasts. This is because if $j$ is not able to complete a broadcast, then it must be a faulty node. Consequently, the $k$-th ``\textit{completed}''  original broadcast made by node $j$ must also be the $k$-th original broadcast by $j$. Hence, Equation (\ref{eq:ri}) still holds for a faulty $j$.


    Define $t^* = \min\{t: \sum_{i=1}^t r_i \ge \frac{1}{4}\}$. Then we have
    \begin{equation}
        \PP\{T > t^*\} = \prod_{i=1}^{t^*} (1-r_i) \le \exp \left( - \sum_{i=1}^{t^*} r_i \right) \le \exp(-1/4).
        \label{eq:firstmover-1b}
    \end{equation}
    

    Define $\calA_p'$ as the nodes in $\calA_p$ that have completed at least one original dummy broadcast among the first $T-1$ completed original dummy broadcasts in phase $p$. In other words, $j\in \calA_p'$ iff $k_j\ge 1$. Then we can derive the following equality, based on the nodes that have made the completed original broadcast(s): 
    \begin{equation}
        \sum_{i=1}^{T-1} r_i = \sum_{j\in \calA_p'} \sum_{k=1}^{k_j} \frac{2^{k-1}}{2n'} = \sum_{j\in \calA_p'} \frac{2^{k_j}-1}{2n'}.
        \label{eq:firstmover-1a}
    \end{equation}
    
    The first equality follows from the definition that the first $T-1$ broadcasts are all dummy, and thus $r_i$ must ``correspond'' to the $k$-th completed original broadcast (for some $1 \leq k \leq k_j$) by some node $j$, whose prior broadcasts are all dummy as well. Furthermore, the $k_j$-th completed original dummy broadcast is the last one by node $j$ (among the first $T-1$ broadcasts across the system). 
    Note that by definition, $r_i$ is a constant for all $i$. However, the summation $\sum_{i=1}^{T-1}r_i$ is indeed a random variable whose randomness comes from each coin flip. This explains why the first equality is valid.
    
    Next, we upper bound the probability that there are multiple original \textit{coin} broadcasts in one phase.
    Note that every active node in $\calA_p$ can make \textit{at most one} original coin broadcast in phase $p$ because a node that makes a original coin broadcast must receive that coin broadcast from itself and thus terminate Algorithm \ref{alg:MAC-FirstMover}. Since by definition, the $T$-th completed original broadcast is the first completed original coin broadcast in the entire system, any original coin broadcast made by some node $j\in\calA_p$ must be the $(k_j+1)$-th original broadcast by node $j$. Equation \eqref{eq:ri} implies that the probability of the $(k_j+1)$-th original broadcast being a coin broadcast is $\frac{2^{k_j}}{2n'}$. 
    
    Let $E_p$ denote the event that there are \underline{strictly more than one original coin broadcast} \underline{in phase $p$} -- these coin broadcasts may or may not be successful. Let $E_p^c$ denote its complement. By union bound, we have
    \begin{align*}
        \PP\{E_p\} 
        &\le \sum_{j\in \calA_p} \PP\{\text{node $j$ makes an original coin broadcast}\} = \sum_{j\in\calA_p} \frac{2^{k_j}}{2n'}.
    \end{align*}
    Consequently, by Equation \eqref{eq:firstmover-1a}, the definition of $t^*$ such that $\sum_{i=1}^t r_i < \frac{1}{4}$ for all $t<t^*$, and the assumption that $n' \ge n \ge |\calA_p|$, we have
        \begin{align*}
        \PP\{ E_p \,|\, T\le t^*\} 
        &\le \sum_{j\in\calA_p} \frac{2^{k_j}}{2n'}  
        = \sum_{j\in\calA_p'} \frac{2^{k_j}}{2n'} + \sum_{j\in \calA_p - \calA_p'} \frac{1}{2n'} \tag{$k_j=0$ for $j\notin \calA_p'$}\\
        &= \left(\sum_{j\in \calA_p'} \frac{2^{k_j}-1}{2n'} +  \sum_{j\in \calA_p'} \frac{1}{2n'}\right) + \sum_{j\in \calA_p-\calA_p'}\frac{1}{2n'}\\
        &= \sum_{i=1}^{T-1}r_i + \sum_{j\in \calA_p}\frac{1}{2n'} 
        = \sum_{i=1}^{T-1}r_i + \frac{|\calA_p|}{2n'} \le \frac{3}{4}.
    \end{align*}
    By Remark \ref{rmk:broadcast-1}, $T\le t^*$, which denotes the event that there is at least one completed original coin broadcast in the first $t^*$ completed original broadcasts, implies that there is at least one \textit{successful} original broadcast in the first $t^*$ completed original broadcasts. Therefore, the fact that $T\le t^*$ together with $E_p^c$ is a \textit{subset} of the events that there is exactly one successful original coin broadcast in phase $p$. Consequently, together with Equation \eqref{eq:firstmover-1b}, we have
    \begin{align*}
        &\PP\{\text{exactly one successful original coin broadcast in phase $p$}\} \\
        &\ge \PP\{E_p^c, T\le t^*\} \ge (1-\exp(-1/4))(1-3/4) \ge 0.05.
    \end{align*}
    This combined with Claim \ref{claim:exactly-one-coin-broadcast} prove the lemma.
\end{proof}

Define the constant $c$ as follows: $c=\frac{\ln(2/\delta)}{0.05}$.
Using $c$ in MAC-FirstMover (Algorithm \ref{alg:MAC-FirstMover}), we can derive the following theorem. The full proof is presented in Appendix \ref{app:theorem:RBC-2}. Roughly speaking, nodes need $O(\log n)$ phases to have a large enough estimated system size $n'$. After that, nodes need a constant number of phases to reach agreement and terminate, due to Lemma \ref{lemma:agreement-firstMover}. Next, Lemma \ref{lemma:termination-firstMover} states that each phase needs $O(n)$ broadcasts on expectation. These give us the desired result. 

\begin{theorem}
    \label{theorem:RBC-2}
    With probability $\ge 1-\delta$, MAC-RBC2 terminates and achieves agreement using $O(n\log n)$ total broadcasts across the entire system.
\end{theorem}


\newpage
\appendix
\section{Correctness Proof of MAC-AdoptCommit}
\label{app:adoptCommit}

\begin{theorem}
MAC-AdoptCommit is correct for binary inputs.
\end{theorem}

\begin{proof}
MAC-AdoptCommit satisfies validity, because $v_i$ is either an input at node $i$ or a value from $proposal_i$, which must be an input from another node.

MAC-AdoptCommit satisfies termination, because all the steps are non-blocking.

MAC-AdoptCommit satisfies coherence. Suppose node $i$ outputs $(commit, v)$ at time $T_3$, and completes line 5 at $T_2$, and line 4 at $T_1$ such that $T_1 < T_2 < T_3$.  

We first make the following observation, namely \textit{Obs1}, no node with input $-v$ has completed line 1 at any time $\leq T_2$. Suppose node $j$ has input $-v$.
By Remark \ref{remark1} in Section \ref{sec:message-processing}, before node $i$ starts to execute line \ref{line:decide}, its message handler has processed all the messages received by the abstract MAC layer. Therefore, the fact that $seen_i[-v]=\text{false}$ at time $T_2$ implies that node $i$ has \textit{not} receive any message of the form $(\texttt{VALUE}, -v)$ at time $T_2$. Consequently, node $j$ has not completed \texttt{mac-broadcast}(\texttt{VALUE}, $-v$) (line 1) at time $T_2$.

Consider the time $T$ when the first \texttt{mac-broadcast}(\texttt{VALUE}, $-v$) is completed (if there is any). At time $T$, any node $k$ that has not crashed yet must have already received (\texttt{PROPOSAL}, $v$) at some earlier time than $T$, because (i) \textit{Obs1} implies that $T > T_2$; and (ii) by time $T$, node $i$ has already completed line 4 (which occurred at time $T_1$). Consider two cases: 

\begin{itemize}
    \item $k$ executes line 2 after receiving (\texttt{PROPOSAL}, $v$): in this case, $k$ sets $proposal_k$ to value $v$ before executing line 3 (potentially at some later point that $T$).
    
    \item $k$ executes line 2 before receiving (\texttt{PROPOSAL}, $v$): in this case: $k$'s input must be $v$; otherwise, $T$ cannot be the first \texttt{mac-broadcast}(\texttt{VALUE}, $-v$) that is completed. (Observe that by assumption of this case, $k$ executes line 2 before node $i$ completes its line 4 at time $T_1$.) 
\end{itemize}
In both cases, 
at line 4, node $k$ can only \texttt{mac-broadcast}(\texttt{PROPOSAL}, $v$). That is, no \texttt{mac-broadcast}(\texttt{PROPOSAL}, $-v$) is possible. Consequently, coherence is satisfied.

MAC-AdoptCommit satisfies convergence. If all the inputs are $v$, then the only value that can appear in $proposal_i$ is $v$ for each node $i$. Moreover, none of the nodes would broadcast $-v$; hence, $seen_i[-v]$ will always be false. Consequently, all nodes would output $(commit, v)$.
\end{proof}

\section{Proof of Theorem \ref{lem:MAC-RBC-termination}}
\label{app:MAC-RBC-termination}



\begin{proof}
Recall that we assume the message oblivious adversary; hence, termination proof is more straightforward. This is because if no node outputs a value, then all nodes rely on the conciliator (flipping a local coin) to reach the same states for the next phase. By construction, nodes may (i) jump to a higher phase with a copied state, (ii) obtain a state that is equivalent to the proposed value from a $\texttt{PROPOSAL}$ message, or (iii) choose its new state randomly. Therefore, there is a non-zero probability that all of these random choices equal to the unique state value obtained using approach (i) or (ii). The reason that these obtained states are identical is due to the \textit{coherence} property of the adopt-commit object (as proved in Appendix \ref{app:adoptCommit}).

In the worst case, all nodes ``move in sync,'' i.e., they enter the same phase concurrently without using the jump, and have their states randomly generated.  Otherwise if there is some ``fast'' node that is in a higher phase, it may force all other nodes to jump to its state after it becomes the ``proposer'' at line 6. We denote the probability that all states are equal after flipping a local coin by $r^*$. Clearly, $r^*=2^{-(n-1)}>0$. Let $P$ be the random variable that denotes the termination phase of MAC-RBC, and note that $P>p$ only if the states are not equal in the first $p$ rounds. Therefore, $\mathbb{P}\{P>p\} \le (1-r^*)^p$. Finally, we conclude the proof by showing that for all $p\ge \ln(1/\delta)/r^* = 2^{n-1}\ln(1/\delta)$,
\begin{equation*}
    (1-r^*)^p \le (1-r^*)^{\ln(1/\delta)/r^*} \le \exp(-\ln(1/\delta)) = \delta.
\end{equation*}
The inequality follows from the identity that $1-x\le \exp(-x)$ for all $x>0$.
\end{proof}

\section{MAC-RBC2}
\label{app:algorithm-MAC-RBC2}

\begin{algorithm*}
\caption{MAC-RBC2 Algorithm: Steps at each node $i$ with input $x_i$}
\label{alg:MAC-RBC2}
\begin{algorithmic}[1]
\footnotesize
\item[{\bf Local Variables:} /* These variables can be accessed and modified by any thread at node $i$. */]{}
	
\item[] $seen_i[0]$ \Comment{(Boolean, phase), initialized to $(false, 0)$}
\item[] $seen_i[1]$ \Comment{(Boolean, phase), initialized to $(false, 0)$}
\item[] $seen_i^2[0]$ \Comment{(Boolean, phase), initialized to $(false, 0)$}
\item[] $seen_i^2[1]$ \Comment{(Boolean, phase), initialized to $(false, 0)$}
\item[] $v_i$ \Comment{state, initialized to $x_i$, the input at node $i$}
\item[] $p_i$ \Comment{phase, initialized to $0$}
\item[] $proposal_i$ \Comment{(value, phase), initialized to $(\perp,0)$}
\item[] $n_0$ \Comment{an initial guess of system size, initialized to some constant natural number}
\item[] $n'$ \Comment{estimated system size, initialized to $1$}
\item[] $c$ \Comment{a constant defined as $c=\frac{\ln(2/\delta)}{0.05}$ 
}
\item[] $coin_i$ \Comment{(Boolean, phase), initialized to $(\perp, -1)$}
\vspace{-5pt}
\item[] \hrulefill    
\vspace{-15pt}
\begin{multicols}{2}
    \State $\texttt{mac-broadcast}(\texttt{ID}, i)$
    \While{true}
        \State $p_{old} \gets p_i$
        \State $\texttt{mac-broadcast}(\texttt{VALUE}, v_i, p_i)$
           
            \If{$proposal_i.phase \geq p_i$}
                \State $(v_i, p_i) \gets proposal_i$
            \EndIf
            \State $\texttt{mac-broadcast}(\texttt{PROPOSAL}, v_i, p_i)$
            
            \If{$p_{old} \neq p_i$}
                \State \textbf{go to} line 2 \Comment{``Jump'' to $p_i$}
            \ElsIf{$seen_i[-v_i].phase < p_i$} 
                \State \textbf{output} $v_i$
            
                \State \texttt{mac-broadcast}($\texttt{VALUE}^2, v_i, p_i$)
                \If{$seen_i^2[-v_i].phase > p_i$}
                    \State $(v_i, p_i) \gets (-v_i, seen^2_i[-v_i].phase)$
                    \State \textbf{go to} line 2 \Comment{``Jump'' to $p_i$}
                \ElsIf{$seen_i^2[-v_i]=(true, p_i)$}
                
                        \State $\slash\slash$ \textit{MAC-FirstMover}
                        \State $n' \gets 2^{\lfloor\frac{p_i}{c}\rfloor} n_0$
                        \State $k \gets 0$
                        \While{$coin_i.phase < p_i$}
                            \If{a local random number $< \frac{2^k}{2n'}$}
                                \State \texttt{mac-broadcast}($\texttt{COIN}, v_i, p_i$)
                            \Else
                                \State \texttt{mac-broadcast}($\texttt{DUMMY}$)
                            \EndIf
                            \State $k \gets k+1$
                        \EndWhile
                    \State \texttt{mac-broadcast}($\texttt{COIN}, v, p$) 
                    \State $(v_i, p_i) \gets coin_i$
                \EndIf    
                \State $p_i \gets p_i + 1$ \Comment{``Move'' to $p_i$}
            \EndIf
    \EndWhile
    
    \columnbreak
    
    \item[$\slash\slash$ \textit{Background message handlers}]
    \Upon{receive($\texttt{VALUE}, v, p$)} 
        \If{$p \geq seen_i[v].phase$}
            \State $seen_i[v] \gets (true, p)$
        \EndIf    
    \EndUpon
    \item[]
    \Upon{receive($\texttt{VALUE}^2, v, p$)} 
        \If{$p \geq seen_i^2[v].phase$}
            \State $seen_i^2[v] \gets (true, p)$
        \EndIf    
    \EndUpon
    
    \item[]
    \Upon{receive($\texttt{PROPOSAL}, v, p$)} 
        \If{$p \geq proposal_i.phase$}
            \State $proposal_i \gets (v, p)$
        \EndIf
    \EndUpon
    \item[]
    
    \item[$\slash\slash$ \textit{Message handlers for MAC-FirstMover}]
    \Upon{receive($\texttt{COIN}, v, p$)} 
        \If{$p = p_i$ ~~\textbf{and}~~$p > coin_i.phase$}
            \State $coin_i \gets (v, p)$
        \ElsIf{$p > p_i$}
            \State $(v_i, p_i) \gets (v, p+1)$
            \State \textbf{go to} line 2  \Comment{``Jump'' to $p_i$}
        \EndIf
    \EndUpon
    
    \item[]
    \Upon{receive($\texttt{DUMMY}$)} 
        \State \textbf{do} nothing
    \EndUpon
\end{multicols}
\end{algorithmic}
\end{algorithm*}

We can get rid of the $coin$ variable and directly use $v_i$ and $p_i$. However, we choose to reserve the variable so that it is more obvious how MAC-RBC2 utilizes MAC-FirstMover. 

The reasons that we need to have the condition $p > coin_i.phase$ are: (i) $coin_i.phase$ may be decoupled from $p_i$; and (ii) each node $i$ has at most two coin broadcasts for a phase $p$.

\section{Proof of Claim \ref{claim:one-coin-broadcast}}
\label{app:claim-one-coin-broadcast}

\begin{proof}[Proof of Claim \ref{claim:one-coin-broadcast}]
    Let $m = (\texttt{COIN}, v, p)$ be a successful coin broadcast in phase $p$. Recall that $m$ is successful because there exists a node $j$ that completes the follow-up broadcast with $(\texttt{COIN}, v, p)$ at some time $t$. 
    Now, consider three groups of nodes:

    \begin{itemize}
        \item For any node $i$ that was in $\calA_p$ before time $t$: $i$ completes MAC-FirstMover for phase $p$ after receiving and processing $m$ or $j$'s follow-up broadcast.

        \item For any node $i$ that has not executed MAC-FirstMover of phase $p$ by time $t$: $i$ would ``jump'' to phase $p$ after receiving and processing $m$ or $j$'s follow-up broadcast.

        \item For any node $i$ that has already completed MAC-FirstMover of phase $p$ before time $t$: this is trivial. Note that this is possible if $i$ processes message(s) faster than $j$ does, or there is a coin broadcast other than $m$. \qedhere
    \end{itemize}
\end{proof}

\section{Proof of Claim \ref{claim:exactly-one-coin-broadcast}}
\label{app:claim:exactly-one-coin-broadcast}

\begin{proof}[Proof of Claim \ref{claim:exactly-one-coin-broadcast}]
In the framework of \cite{RandConsensus_Aspnes_DC2012}, if every node that has not crashed obtains the same output from the conciliator object, then all the fault-free nodes are guaranteed to terminate in the next phase. This design, the definition of a successful coin broadcast, and the ability to jump to a higher phase in MAC-RBC2 imply the claim. This is because for all nodes that update its state $v_i$ in phase $p$, they must use the same outcome from the conciliator object (the value field of the successful coin broadcast). For the other nodes that jump to phase $p+1$ (from a phase $< p$), they must either receive phase-$p$ coin broadcast(s) or receive the messages from the adopt-commit object in phase $p+1$. These messages and phase-$p$ coin broadcasts (both the one and only original coin broadcast and follow-up coin broadcasts) must contain exactly the same value. 
\end{proof}

\section{Proof of Theorem \ref{theorem:RBC-2}}
\label{app:theorem:RBC-2}

\begin{proof}
    First, we can decompose the total number of broadcasts by all fault-free nodes, denoted by $N$, into three components $N=N^{RBC}+N^{O}+N^{F}$, where (i) $N^{RBC}$ denotes the number of broadcasts required by the part of adopt-commit (i.e., all the communication in Algorithm \ref{alg:MAC-RBC}); (ii) $N^{O}$ denotes the number of original broadcasts used in MAC-FirstMover for all phases;  and (iii) $N^{F}$ denotes the number of follow-up broadcasts used in MAC-FirstMover for all phases. 
    
    Let $P$ denote the random variable of the first phase index in which the agreement is achieved, i.e., all nodes that have not crashed begin with same $v$ in this phase.
    
    First observe that in each phase, each node makes $O(1)$ broadcasts for adopt-commit and one  follow-up broadcast in for MAC-FirstMover. Therefore, $N^{RBC}+N^F = O(n P)$. The rest of the proof focuses on bounding $N^O$.
    
    Let $n'_p = 2^{\lfloor p/c\rfloor} n_0$ denote the input to MAC-FirstMover, namely the estimated system size in phase $p$. Then $n'_p \ge n$ for all $p\ge c(1+\log_2(n/n_0))$. Therefore, Lemma \ref{lemma:agreement-firstMover} implies that the event $E_p$ of no agreement in phase $p$ has bounded probability $\PP\{E_p\} \le 1-0.05$ for all $p\ge c(1+\log_2(n/n_0))$. Consequently,
    \begin{align*}
        \PP\{P > c(1+\log_2(n/n_0)) + q\} 
        \le \prod_{i=1}^q \PP\{E_{\lfloor c(1+\log_2(n/n_0) \rfloor+i}\}
        \le \prod_{i=1}^q (1-0.05) \le \exp(-0.05q).
    \end{align*}
    Consequently, let $p^*=c(1+\log_2(n/n_0)) + \frac{\ln(2/\delta)}{0.05}$. Upon substituting $q=\frac{\ln(2/\delta)}{0.05}$ into the previous bound, we have
    \begin{equation}
        \PP \left\{ P > p^* \right\} \le \delta/2.
        \label{eq:RBC2-agreement}
    \end{equation}
    Note that with $c=\frac{\ln(2/\delta)}{0.05}$, $p^* = \frac{\ln(2/\delta)}{0.05}(2+\log_2(n/n_0)) = O(\ln(n)\ln(1/\delta))$.
    
    Let $N^O_p$ denote the number of original broadcasts made byfault-free nodes in MAC-FirstMover of phase $p$. Lemma \ref{lemma:termination-firstMover} implies that with probability $\ge 1-\delta$, $N^O_p \le 2n'_p\ln(1/\delta)$. 
    Therefore, upon applying union bound, we have with probability $\ge 1-\delta/2$, 
    \begin{align*}
        \sum_{p=1}^{p^*} N^O_p 
        &\le \sum_{p=1}^{p^*} 2n'_p\ln(2p^*/\delta) \tag{recall that  $n_p'=2^{\lfloor p/c\rfloor}n_0$}\\
        &\le 2cn_0 \ln(2p^*/\delta) \sum_{q=1}^{\lceil p^*/c\rceil} 2^p \\
        &\le 4 cn_0\ln(2p^*/\delta) 2^{\lceil p^*/c\rceil} \tag{substitute definition of $p^*$}\\
        &= 4n_0 \frac{\ln(2/\delta)}{0.05} \ln\left( \frac{2\ln(2/\delta)(2+\log_2(n/n_0))}{0.05\delta} \right) \exp_2\left( 2+\log_2(n/n_0) \right) \\
        &= 320 n \ln(2/\delta) \ln \left( \frac{2\ln(2/\delta)(2+\log_2(n/n_0))}{0.05\delta} \right) \\
        &= O\left( n\ln(1/\delta)\ln\left(\frac{\ln(n)\ln(1/\delta)}{\delta}\right)\right).
    \end{align*} 
    Equivalently, we have
    \begin{equation}
        \PP\left\{ \sum_{p=1}^{p^*} N^O_p > 320 n \ln(2/\delta) \ln \left( \frac{2\ln(2/\delta)(2+\log_2(n/n_0))}{0.05\delta} \right) \right\} \le \delta/2.
        \label{eq:RBC2-termination}
    \end{equation}

    Upon combining Equations \eqref{eq:RBC2-agreement}, \eqref{eq:RBC2-termination} and applying union bound, we have with probability $\ge 1-\delta$, MAC-RBC2 achieves agreement (and thus termination) with
    \begin{align*}
        N &= N^{RBC} + N^O + N^F
        = O\left( n\ln(n)\ln(1/\delta)\right)
        \qedhere
    \end{align*}

    
\end{proof}
\section{Impossibility of $(n-1)$-set consensus}
\label{s:set-consensus}

In $k$-set consensus \cite{welch_book,AA_nancy}, nodes can have at most $k$ different outputs and the output must be some node's input. Formally, the $k$-set consensus needs to 

\begin{itemize}
    \item \textit{$k$-Agreement}: The cardinality of the set of outputs is at most $k$.

    \item \textit{Validity}: Every output is an input value of a node.
\end{itemize}

Obviously, $n$-set consensus is trivial for $n$ nodes, since each node can simply decide on its input. We prove that $(n-1)$-set consensus is impossible to implement in the abstract MAC layer if up to $n-1$ nodes may fail, i.e., there is no wait-free algorithm for solving $(n-1)$-set consensus. 

\subsection{Preliminaries}

We introduce some useful terminology and concept to facilitate our impossibility proof. 

An \textit{execution} of an algorithm is a finite sequence of nodes' communication steps (i.e., invoking \texttt{mac-broadcast}, receiving a message, or receiving an acknowledge of a \texttt{mac-broadcast} invoked earlier). In each execution, each node $i$ is given an input value, denoted by $x_i$. Then $i$ performs a series of computation and communication, and finally terminates with an output value. 

An execution $\alpha$ is said to be an execution of the set of nodes $P$ if \textit{all} nodes in $P$ do some communication steps in $\alpha$, and only these nodes. $P$ is the called the ``\textit{participating node}'' set of $\alpha$. When the set is clear from the context, we often refer $P$ simply as nodes. 

For the impossibility proof, we only care about when all participating nodes output some value. Note that they may take steps after all the other nodes terminate. 

Two executions $\alpha, \alpha'$ are ``\textit{indistinguishable}'' to node $i$, denoted as $\alpha \indistinguish{i} \alpha'$, if the state the state of $i$ after both executions is identical. $\alpha \indistinguish{P} \alpha'$, if $\alpha \indistinguish{i} \alpha'$ for each $i \in P$.

Given an execution $\alpha$, a node $i$ is \textit{unseen}, if its communication steps in $\alpha$ \textit{only} occur \textit{after} all other nodes have terminated. Note that we are only concerning about fault-free node $i$ for the impossibility proof. For node $i$ that is \textit{not unseen}, it is called \textit{seen} in $\alpha$. 


We first prove the following important lemma.

\begin{lemma}
    \label{lemma:seen}
    Consider a set of nodes $P$ where $i \in P$. If node$i$ is seen in an execution $\alpha$ of $P$, then there is an  execution $\alpha'$ of $P$ such that 

    \begin{itemize}
        \item $\alpha' \neq \alpha$; 
        \item $\alpha' \indistinguish{P-\{i\}} \alpha$, i.e., for nodes other than $i$, the two executions are indistinguishable; and

        \item Node $i$ is also seen in $\alpha'$.
    \end{itemize}
\end{lemma}

\begin{proof}
    The lemma follows from the observation that we can create the execution $\alpha'$ by shuffling the broadcast or the acknowledge event at node $i$ without affecting the state at nodes $P - \{i\}$. For example, suppose that node $i$'s acknowledgement occurs after some node $j$ broadcasts a value. Then, we can create $\alpha'$ by moving the acknowledgement before $j$'s broadcast event. In this case, $i$ is still seen by node $j$, but node $j$ canno distinguish between the two executions. 
\end{proof}

\subsection{Impossibility Proof}

Next we prove the impossibility by contradiction. That is, we assume that there is a wait-free algorithm in the abstract MAC model for solving $(n-1)$-set consensus, which we will derive a contradiction at the end of this section. 

First, we introduce an important definition:

\begin{definition}
    For set of nodes $P$, denote by $C_m$ (where $1 \leq m \leq n$) the set of all executions in which

    \begin{itemize}
        \item Only the first $m$ nodes take communication steps (and the others are assumed to be crashed);

        \item Each node $i \in P$ has input value $i$; and

        \item All the values from $1$ to $m$ are some node's output. (That is, there are $m$ distinct output values.)
    \end{itemize}
\end{definition}

To complete the impossibility proof, we will prove that $C_n$ is non-empty, i.e., there exists at least one execution in which all $n$ values are some node's output, which implies that $(n-1)$-set consensus is infeasible.

\begin{lemma}
    \label{lemma:odd-Cm}
    For every $m~~(1 \leq m \leq n)$, the cardinality of $C_m$ is odd.
\end{lemma}

\begin{proof}
    The proof is by induction on $m$. The base case is when $m = 1$. Consider a solo execution of node $1$, i.e., only node $1$ takes communication steps. Since the algorithm is wait-free, node $1$ outputs a value after taking a finite positive number of communication steps. By the validity property, node $1$ outputs $1$, so there is a unique execution in $C_1$. Hence, $|C_1| = 1$. Here we use $|\cdot|$ to denote the size of a set.

    Assume by induction that the claim of the lemma holds for some $m, 1 \leq m < n$. Let $X_{m+1}$ be the set of all tuples of the form $(\alpha, i), 1 \leq i \leq m+1$, such that $\alpha$ is an execution where only nodes $1$ to $m+1$ take communication steps, and all the values from $1$ to $m$ are the output at some node other than $i$ in $P$, where $P$ is the set of participating nodes in $\alpha$ Node $i$ decides on some arbitrary value. Note that node $i$ may or may not output $m+1$ in an execution.

    We next prove the following claim.

    \begin{claim}
        \label{claim:same-parity}
        $|X_{m+1}|$ and $|C_{m+1}|$ have the same parity.
    \end{claim}

    \begin{proof}[Proof of Claim \ref{claim:same-parity}]

    Let $X_{m+1}'$ be the subset of $X_{m+1}$ which contains all tuples $(\alpha, i)$, such that all values from $1$ to $m+1$ are some node's output in $\alpha$. 
    We will show that the size of $X_{m+1}'$ is equal to the size of $C_{m+1}$. 

    \begin{itemize}
        \item \textit{Observation 1}: $(\alpha, i)$ in $X_{m+1}'$ if and only if $\alpha$ is in $C_{m+1}$.

        We first argue that if $(\alpha, i)$ is in $X_{m+1}'$, then $\alpha$ is in $C_{m+1}$. Since $m+1$ values are outputs at $m + 1$ nodes in $\alpha$, node $i$ is the only node that outputs $m+1$ in $\alpha$. Consequently, there is no other tuple $(\alpha, j)$ for $j \neq i$ in $X_{m+1}'$ with the same execution $\alpha$. 
        We now show the other direction. Consider the case when $\alpha$ is an execution in $C_{m+1}$. Then observe that in $\alpha$, $m + 1$ values are outputs at $m+1$ nodes, and there is a unique node $i$ which outputs $m+1$ in $\alpha$. Hence, $\alpha$ appears in $X_{m+1}'$ exactly once, in the tuple $(\alpha, i)$.

        \item \textit{Observation 2}: $X_{m+1} \setminus X_{m+1}'$ contains an even number of tuples.

        If $(\alpha, i)$ is a tuple in $X_{m+1}$, but not in $X_{m+1}'$, then node $i$ outputs some value $v$ other than $m+1$ in $\alpha$. Since $(\alpha, i)$ is in $X_{m+1}$, all values from $1$ to $m$ are outputs at some node other than node $i$ in $\alpha$. In other words, there is a unique node $j \neq i$ that also outputs $v$ in $\alpha$. This is because by assumption, there are, in total, $m$ different output values in $X_{m+1}$, if node $i$ does not output $m+1$.
        It follows that both $(\alpha, i)$ and $(\alpha, j)$ are in $X_{m+1} \setminus X_{m+1}'$. Furthermore, these two tuples are the only appearances of $\alpha$ in $X_{m+1}$. Therefore, $X_{m+1} \setminus X_{m+1}'$ contains an even number of tuples.
    \end{itemize}
    
     Observation 2 implies that the sizes of $X_{m+1}$ and $X_{m+1}'$ have the same parity. This together with Observation 1 indicate that the sizes of $X_{m+1}$ and $C_{m+1}$ have the same parity, proving the claim.
    \end{proof}

    To complete the proof of the lemma, we will argue that $|X_{m+1}|$ is odd. We will partition the tuples $(\alpha, i)$ in $X_{m+1}$ into three disjoint subsets:

    \begin{itemize}
        \item \textbf{Node $i$ is seen in $\alpha$}: 
        
        By Lemma \ref{lemma:seen}, for each execution $\alpha$ in which only nodes $1$ to $m+1$ take communication steps, and node $i$ is seen, there is a unique execution $\alpha' \neq \alpha$ with (i) nodes $1$ to $m+1$ taking communication steps; (ii) $i$ is seen, and (iii) all nodes other than $i$ output the same values as in $\alpha$. Hence, the tuples in $X_{m+1}$ in which $i$ is seen in the execution can be partitioned into disjoint pairs of the form $\{(\alpha, i), (\alpha', i)\}$. This indicates that there is an even number of such tuples (that correspond to the case when node $i$ is \textbf{seen} in $\alpha$).

        \item \textbf{Node $i$ is unseen in $\alpha$ and $i \neq m+1$}:

        By definition, $i \in \{1, \cdots, m\}$ and all values from $1$ to $m$ are the output at some node other than $i$ in $\alpha$. Thus, we have that the value $i$ is the output in $\alpha$ by some node $j \neq i$. However, $i$ is \textbf{unseen} by assumption of this case. Node $j$ must have output the value $i$ before node $i$ introduces $i$ as an input value. 
        
        Now, consider an execution $\alpha'$ which has the same prefix (of steps) as $\alpha$ except that $\alpha'$ does not contain the communication steps by node $i$ at the end. (We know that these steps are at the end, because node $i$ is unseen.) In this execution $\alpha'$, the algorithm violates the validity property, because no node has the input value $i$ in $\alpha'$. This implies that there is no such tuple in $X_{m+1}$.

        \item \textbf{Node $i$ is unseen in $\alpha$ and $i = m+1$}:

        We show a bijection between this subset of $X_{m+1}$ and $C_m$. Since node $i = m+1$ is \textit{unseen} in $\alpha$, in $\alpha$, we must have all nodes other than $m+1$ take communication steps and output all values from $1$ to $m$ in $\alpha$ and then at the end of the execution, after all other nodes terminate, node $m+1$ takes communication steps \textit{alone}. 
        
        Consider the execution $\alpha'$ which has the same prefix (of steps) as $\alpha$ except that $\alpha'$ does \textit{not} contain the communication steps by node $m+1$ at the end. Observe that by definition, $\alpha'$ is in $C_m$. 
        
        On the other direction, every $\alpha' \in C_m$ can be extended to an execution $\alpha$ by appending singleton communication steps of node $m+1$ to its end such that $(\alpha, m+1)$ is in $X_{m+1}$ and node $m+1$ is \textit{unseen} in $\alpha$.
        
        By the induction hypothesis, the size of $C_m$ is odd, so the bijection implies that $X_{m+1}$ has an odd number of tuples $(\alpha, m)$ in which node $m$ is \textit{unseen} in $\alpha$.
    \end{itemize}

    These three cases imply that the size of $X_{m+1}$ is odd. It follows from Claim \ref{claim:same-parity} and this observation that $|C_{m+1}|$ is an odd number, proving the lemma.
\end{proof}

Plugging $m = n$ into Lemma \ref{lemma:odd-Cm}, we have that $|C_n|$ is odd. In other words,  the size is non-zero. Therefore, we know that there must exist an execution in which all $n$ values are output values at some node, a contradiction. Formally, we prove the following impossibility:

\begin{theorem}
    In the abstract MAC layer with $n$ nodes, it is impossible to implement a wait-free $(n-1)$-set consensus algorithm.
\end{theorem}

\section{Anonymous Storage-Efficient Approximate Consensus}
\label{app:approximate}


\vspace{3pt}
\noindent\textbf{Approximate Consensus.}
In this section, we study approximate consensus, which is an important primitive related to many other fundamental problems, such as data aggregation, decentralized estimation, clock synchronization and flocking \cite{vaidya_PODC12,Tseng_SIROCCO16}. Our first algorithm MAC-AC has constant storage complexity and achieves convergence of $1/2$. The second algorithm MAC-AC2 requires the knowledge of the upper bound on $n$ and has a worse convergence rate; however, it uses only two values and one Boolean, which we conjecture to be \textit{optimal}. Both algorithms are \textit{anonymous} and \textit{wait-free}.

\begin{definition}[Approximate Consensus \cite{AA_Dolev_1986}]
    A correct approximate consensus algorithm  needs to satisfy: (i) \textit{Termination}: Each fault-free node decides an output value within a finite amount of time; (ii) \textit{Validity}: Each output is within the range of the inputs; and (iii) \textit{$\epsilon$-agreement}: The outputs of all fault-free nodes are within $\epsilon$.
\end{definition}

Following the literature \cite{AA_nancy,welch_book}, we assume that the range of initial inputs of all nodes is bounded. This is  typical in many connected devices applications that are sensing nearby environments. For example, inputs could be temperature, air pollution level, target location, etc. 
Without loss of generality, we scale the inputs to $[0,1]$.

\subsection{Algorithm MAC-AC and Correctness Proof}

We first present MAC-AC in Algorithm \ref{alg:MAC-AC}. The algorithm adopts the iterative structure from prior works, e.g.,  \cite{AA_Dolev_1986,vaidya_PODC12}. Nodes proceed in phases, and their state is attached with a phase index. Node $i$ initializes phase index $p_i=0$ and state value $v_i=x_i$, its input in $[0, 1]$. To reduce storage space, each node does not store all the received messages as in prior algorithms. Instead, node $i$ only stores two extra states: the maximum and minimum received states, denoted by $v_{max}$ and $v_{min}$, respectively. These two values are only updated by the message handler. 

In prior algorithms (e.g.,  \cite{AA_Dolev_1986,vaidya_PODC12}), nodes always proceed phase by phase, which requires nodes in phase $p$ to store messages from higher phases so that slow nodes could eventually catch up. We integrate the  ``jump'' technique \cite{Tseng_DISC20_everyonecrash} to avoid such wasted storage  -- upon receiving a state from a higher phase $p$, a node can simply copy the state and ``jump'' to the phase $p$. 
Node $i$ uses $p_{next}$, $v_{next}$, and a Boolean $jump$ to decide how to proceed to a higher phase.

At the beginning of each phase $p_i$, node $i$ resets $v_{min}, v_{max}$. Then it use the abstract MAC layer to broadcast its current phase index and value (line \ref{line:MAC-AC-mac-broadcast}). Meanwhile, the background message handler may receive and process  messages from other nodes. After node $i$ receives the acknowledgement of its $\texttt{mac-broadcast}$ primitive, it concludes the current phase by deciding how to proceed to a higher phase (line 8-13) with the following two approaches:
\begin{enumerate}
    \item (line 8-10): If node $i$ did not receive any message of the form $(v,p_j)$ with $p_j>p_i$, then it updates to phase $p_i+1$ with state $(v_{min}+v_{max})/2$. In this case, we say node $i$ ``\textit{moves}'' to phase $p_i+1$.
    \item (line 11-12): If node $i$ received a  message with a higher phase, then it copies the received state and phase. 
    In this case, we say node $i$ ``\textit{jumps}'' to $p_j$.
\end{enumerate}
Each node repeats this process until phase $p_{end}$. Observe that messages could arrive and be processed at any point of time. To ensure that $p_i$ and $v_i$ are always in sync, we use locks to protect the reset (line 3 and 4) and update (line 8 to 12) procedures. 

We use a lock to ensure a block of codes are executed atomically. Recall that \texttt{mac-broadcast} is not atomic; hence, the message handler could still process messages during any point of time between line 5 and line 7. 

\begin{algorithm*}
\caption{MAC-AC: Steps at each node $i$ with input $x_i$}
\label{alg:MAC-AC}
\begin{algorithmic}[1]
\footnotesize
\item[{\bf Local Variables:} /* These can be accessed by any thread at node $i$. */]{}
	
\item[] $p_i$ \Comment{phase, initialized to $0$}
\item[] $v_i$ \Comment{state, initialized to $x_i$ (input at node $i$)}
\item[] $jump$ \Comment{Boolean, initialized to $false$}
\item[] $v_{min}, v_{max}$ \Comment{minimum and maximum received states, initialized to $v_i$}

\vspace{-5pt}

\item[] \hrulefill    

\vspace{-15pt}

\begin{multicols}{2}
    \For{$p_i \gets 0$ to $p_{end}$}
        \State \texttt{lock}()
        \State $v_{min} \gets v_i;~~~ v_{max} \gets v_i$
        \State $jump \gets false$
        \State \texttt{unlock}()
        \State \texttt{mac-broadcast}($v_i, p_i$) \label{line:MAC-AC-mac-broadcast}
        \State \texttt{lock()}
        \If{$jump = false$} \label{line:MAC-AC-wait}
            \State $v_i \gets \frac{1}{2}(v_{max} + v_{min})$ 
            \State $p_i \gets p_i + 1$ \Comment{``Move'' to  $p_i$}
            \label{line:MAC-AC-move}
        \Else 
            \State \textbf{go to} line 3 in phase $p_i$ \Comment{``Jump'' to $p_i$}
        \EndIf
        \State \texttt{unlock()}
    \EndFor
    \State output $v_i$
    
    \columnbreak
    
    \item[$\slash\slash$ \textit{Background message handler}] 
    \Upon{receive($v,p$)} 
        \State \texttt{lock()}
        \If{$p > p_i$}
            \State $p_i \gets p$
            \State $v_i \gets v$ 
            \State $jump \gets true$
        \ElsIf{$p = p_i $}
            \If{$v > v_{max}$}
                \State $v_{max} \gets v$
            \EndIf 
            \If{$v < v_{min}$}
                \State $v_{min} \gets v$
            \EndIf 
        \EndIf
        \State \texttt{unlock}()
    \EndUpon
    \Statex
    \Statex
    \Statex
\end{multicols}
\end{algorithmic}
\end{algorithm*}


\noindent\textbf{Correctness Proof.}~Validity and termination are obvious, since we consider only crash faults. We focus on the convergence proof below. 
Consider node $i$ at phase $p$. First, we formalize the definition of phase $p$. In each iteration of the for-loop, phase $p$ of node $i$ starts at line 2 and ends at line 6.  
Although the value of $v_i$ may change during this period, we define the phase-$p$ state of $i$ as $v_i$ \underline{at the beginning of phase $p$}, denoted by $v_i[p]$. 
For completeness, we define $v_i[p]=\perp$ if node $i$ skips phase $p$ (by jumping to a higher phase) or crashes before phase $p$. 
In other words, node $i$ broadcasts $(v_i[p], p)$ using the abstract MAC layer at line 6; and the MAC layer ensures that every node that has not crashed will receive $v_i[p]$ when line 6 completes, even if the broadcaster crashes afterwards. This implies the following remark: 
\begin{remark}
\label{rmk:MAC-AC2}
Every received and processed message is of form $(v_j[p_j],p_j)$ for some node $j$.
\end{remark}

Next, we define $V[p]$ as the multiset of phase-$p$ states of all nodes, excluding $\perp$. For brevity, we define $min_p=\min V[p], max_p=\max V[p]$. Finally, we denote $i_p$ as the first node that completes line 6 in phase $p$. Since it is the first node, its phase index does \textit{not} change during line 2 and line 6. If multiple nodes complete line 6 concurrently, we let $i_p$ be any one of them. Due to the property of \texttt{mac-broadcast}, $v_{i_p}[p]$  is the ``common value'' received by fault-free nodes that we need to prove convergence. We then prove a key lemma.

\begin{lemma}
\label{lem:MAC-AC-move-nodes}
Let $j$ be a node other than $i_p$ that \textbf{moves} to phase $p+1$. Then we have
\[
v_j[p+1] \in \left[\frac{min_p+v_{i_p}[p]}{2}, \frac{max_p+v_{i_p}[p]}{2}\right].
\]
\end{lemma}

\begin{proof}
Define $R_j[p]$ as the multiset of all states in the tuple $(v,p)$ ``{received and \textit{processed}}'' by $j$'s message handler in phase $p$.
Due to our message processing model described in Section \ref{sec:message-processing}, the message handler only processes messages when the main thread is not locked, i.e., after line 5 ends and before line 7 starts. By Remark \ref{remark1}, there is no pending message to be processed in the background handler when line 7 starts. Moreover, the handler does not process any message when the main thread is locked, i.e., during line 7 to 13. Therefore, by Remark \ref{rmk:MAC-AC2}, 
$R_j[p]$ is effectively the multiset of all phase-$p$ states received and processed during the execution of line 5 to 7. Consequently, at line $9$, $v_{max}=\max R_j[p]$ and $v_{min} = \min R_j[p]$;
and $v_j[p+1] = \frac{1}{2}(\min R_j[p]+\max R_j[p])$.

Suppose $i_p$ completed $\texttt{mac-broadcast}$ (line \ref{line:MAC-AC-mac-broadcast}) at time $T$. By the property of the abstract MAC layer, node $j$ has received message $(v_{i_p}[p],p)$ by time $T$. By definition of $i_p$, node $j$ has \textit{not} completed line 6 by time $T$, so it is in phase at most $p$ and the message processing assumption (Remark \ref{remark1}). 
Recall that as mentioned in Section \ref{s:preliminary}, all the messages received by node $j$ by time $T$ will be processed before executing line 7, owing to the way that the message handler works. Therefore, $v_{i_p}[p]\in R_j[p']$ for some $p'\leq p$. Since by assumption, node $j$ \textit{moves} to phase $p+1$, it must be the case that $v_{i_p}[p]\in R_j[p]$.

Finally, by definition,  $R_j[p]\subseteq V[p]$, so $\min V[p]\leq \min R_j[p] \leq \max R_j[p]\leq \max V[p]$. Moreover, $v_{i_p}[p]\in R_j[p]$, so $\min R_j[p] \leq v_{i_p}[p] \leq \max R_j[p]$. In conclusion,
\begin{align*}
\frac{\min V[p] + v_{i_p}[p]}{2} &\leq v_j[p+1] \\
&= \frac{\min R_j[p]+\max R_j[p]}{2} \leq \frac{v_{i_p}[p] + \max V[p]}{2}. 
\end{align*}
\end{proof}

With some standard arithmetic manipulation and Lemma \ref{lem:MAC-AC-move-nodes}, we can obtain the following theorems, which prove the exponential convergence. 

\begin{theorem}
\label{thm:MAC-AC-agreement}
For any $p\geq 0$,
\[
\max V[p] - \min V[p] \leq \frac{\max V[0] - \min V[0]}{2^p}.
\]
\end{theorem}

\begin{proof}
Consider any phase $p+1$. Node $j$ with $v_j[p+1]\neq \perp$ must proceed to phase $p+1$ by either jumping or moving. If it moves to phase $p+1$, then by Lemma \ref{lem:MAC-AC-move-nodes}, 
\[
v_j[p+1] \in \left[\frac{min_p+v_{i_p}[p]}{2}, \frac{max_p+v_{i_p}[p]}{2}\right].
\]
If it jumps to phase $p+1$, then there exists some node $k$ such that $v_j[p+1] = v_k[p+1]$ and node $k$ moves to phase $p+1$. In other words, node $k$ is essentially the ``source node'' where $j$ copies its state from. Then $v_k[p+1]$ also satisfies the equation above. In conclusion, this equation holds for any $v_j[p+1]$ for any node $j$. Therefore,
\[
\max V[p+1] - \min V[p+1] \leq \frac{\max V[p] - \min V[p]}{2}.
\]
We then prove the theorem by unrolling this recursive expression.
\end{proof}

Theorem states that the convergence rate is $1/2$, which immediately follows the following theorem. 

\begin{theorem}
\label{thm:MAC-AC-converge}
Algorithm MAC-AC achieves $\epsilon$-agreement for any $p_{end} \geq \log_2(\frac{1}{\epsilon})$.
\end{theorem}

\subsection{Further Reducing Storage Complexity}

Next we present  MAC-AC2, 
which further reduces storage complexity, assuming the knowledge of an upper bound on $n$. Each node $i$ only keeps three values: state $v_i$, phase $p_i$ and Boolean $jump$. Upon receiving $(v,p)$ from a higher phase $p>p_i$, $i$ jumps to phase $p$ and copies state $v$; and upon receiving $(v,p)$ with  $p=p_i$, its new state becomes the average of its current state and $v$. 

\begin{algorithm}
\caption{MAC-AC2: Steps at each node $i$ with input $x_i$}
\label{alg:MAC-AC2}
\begin{algorithmic}[1]
\footnotesize
\item[{\bf Local Variables:} /* These variables can be accessed and modified by any thread at node $i$. */]{}
	
\item[] $p_i$ \Comment{phase, initialized to $0$}
\item[] $v_i$ \Comment{state, initialized to $x_i$, the input at node $i$}
\item[] $jump$ \Comment{Boolean, initialized to $false$}

\vspace{-8pt}

\item[] \hrulefill

\vspace{-8pt}

\begin{multicols}{2}
    \For{$p_i \gets 0$ to $p_{end}$}
        \State $jump\gets false$
        \State \texttt{mac-broadcast}($v_i, p_i$) 
        \If{$jump = false$}
            \State $p_i\gets p_i+1$ \Comment{``Move'' to $p_i$}
        \Else 
            \State \textbf{go to} line 2 in phase $p_i$ \Comment{``Jump'' to $p_i$}
        \EndIf
    \EndFor
    \State output $v_i$
    
    \columnbreak
    
    \item[$\slash\slash$ \textit{Background message handler}]
    \Upon{receive($v,p$)}
        \If {$p > p_i$}
            \State $p_i\gets p$
            \State $v_i \gets v$
            \State $jump \gets true$
        \ElsIf {$p = p_i$}
            \State $v_i\gets (v_i+v)/2$
        \EndIf
    \EndUpon
\end{multicols}
\end{algorithmic}
\end{algorithm}


This algorithm does not need a lock but requires the knowledge of an upper bound on $n$ to determine $p_{end}$.  Moreover, the convergence rate in the worst case could be as slow as $(1-2^{-n})$. However, we believe that its simplicity makes it appropriate for many small connected devices. 

\noindent\textbf{Convergence Proof.}~ 
Intuitively when a node moves to the next phase, its state is a weighted average of all  received (and processed) states from the same or higher phase. Therefore, the range of all states should shrink strictly. However, the proof requires a  finer-grained setup and more involved arithmetic to obtain the bound and the convergence proof.

We first make the following definitions. Phase $p$ of node $i$ starts at line 2 in Algorithm \ref{alg:MAC-AC2} and ends when line 4 starts. We define the phase-$p$ state of node $i$ as $v_i$ at the beginning of phase $p$, denoted by $v_i[p]$. For completeness, $v_i[p]=\perp$ if $i$ crashes before phase $p$. Next, we define $V[p]$ as the multiset of all phase-$p$ states excluding $\perp$,  and $i_p$ as the \underline{first} node that completes line 2 in phase $p$. 
In the following discussion, we use $r$ for the subscripts of messages, and $i, j$ for nodes. 

\begin{theorem}
\label{thm:mac-ac2}
If node $j$ \textbf{moves} to phase $p+1$, then
\[
v_j[p+1] \in \left[ \min V[p] + \frac{v_{i_p}[p]-\min V[p]}{2^{n}}, \max V[p] - \frac{\max V[p]-v_{i_p}[p]}{2^{n}} \right].
\]
\end{theorem}

\begin{proof}
Define $R_j[p]$ as the multiset of all states in the tuple $(v,p)$ received and \textit{processed} by $j$ in phase $p$. By Remark \ref{rmk:MAC-AC2}, $R_j[p]$ is effectively the multiset of phase-$p$ states received and processed by node $j$. We relabel the values in $R_j[p]$ following the increasing order of the time that each message is processed, i.e., $R_j[p]=\{v_1,\ldots,v_R\}$ where $R=|R_j[p]|$ for simplicity. Since node $j$ moves to phase $p+1$, it does not receive any state from a higher phase. Therefore,
\begin{align}
    v_j[p+1] 
    &= \frac{1}{2}\left( \dotsm \frac{1}{2}\left( \frac{1}{2}\left( v_j[p] + v_1 \right) + v_2 \right) \dotsm + v_R \right) \notag\\
    &= \frac{v_j[p]}{2^R} + \sum_{r=0}^{R-1} \frac{v_{R-r}}{2^{r+1}}. \label{eq:MAC-AC2}
\end{align}
For brevity, the subscript for the equation above is no longer the node index, except for $v_j$ for node $j$. Note that all values in $R_j[p]$, as well as $v_j[p]$, are in $V[p]$. Also there are at most $n$ phase-$p$ states, so $R\leq n$. Finally, $v_{i_p}[p]\in R_j[p]$ because otherwise node $j$ would jump to $i_p$'s state. Without loss of generality, assume that $v_{R-r^*}$ be $v_{i_p}[p]$. In particular, $r^*\leq R-1\leq n-1$. Hence, we have
\begin{align*}
    v_j[p+1] 
    &= \frac{v_j[p]}{2^R} + \left(\sum_{r=0, r\neq r^*}^{R-1} \frac{v_{R-r}}{2^{r+1}}\right) + \frac{v_{R-r^*}-\max V[p] +\max V[p]}{2^{r^*+1}} \\
    &\leq \frac{\max V[p]}{2^R} + \left(\sum_{r=0, r\neq r^*}^{R-1} \frac{\max V[p]}{2^{r+1}}\right) + \frac{\max V[p]}{2^{r^*+1}} - \frac{\max V[p]-v_{i_p}[p]}{2^{r^*+1}} \\
    &\leq \max V[p] - \frac{\max V[p] - v_{i_p}[p]}{2^n}.
\end{align*}
The lower bound can be obtained symmetrically.
\end{proof}

Following the same proof of Theorem \ref{thm:MAC-AC-agreement}, we can prove the convergence rate using Theorem \ref{thm:mac-ac2} in the theorem below, which allows us to derive $p_{end}$. 

\begin{theorem}
\label{thm:MAC-AC2-converge}
For any $p\geq 0$,
$\max V[p] - \min V[p] \leq (\max V[0] - \min V[0])(1-2^{-n})^p.$ 
\end{theorem}

\paragraph{Discussion} We conjecture that MAC-AC2 achieves the optimal storage complexity for achieving approximate consensus. Intuitively, phase is used for determining the termination condition. Suppose nodes do not store messages that they cannot process currently (e.g., due to mismatching phases). Then we believe that the jump mechanism is necessary for nodes to learn states from a higher phase. In this case, the Boolean variable is used to distinguish between moving and jumping so that nodes can keep track of the current phase accurately. 

\end{document}